\documentclass[11pt, letter]{article}
\usepackage[margin=1in]{geometry}

\usepackage{latexsym}
\usepackage{algorithm}
\usepackage{amsmath}
\usepackage{amsfonts}
\usepackage{amssymb}
\usepackage{bbm}
\usepackage{tabls}
\usepackage{theorem}
\usepackage[dvips]{graphicx}
\usepackage{color}
\usepackage{enumerate}
\usepackage{algorithmicx}
\usepackage[noend]{algpseudocode}
\newcommand{\ignore}[1]{}
\usepackage{sansmath}
\usepackage{hyperref}
\usepackage[round]{natbib}
\newcommand{\poly}{\operatorname{poly}}

\newcommand{\Prob}[2]{\underset{#1}{\ensuremath{\mathsf{Pr}}}\left[#2\right]}
\newcommand{\Expectation}[2]{\underset{#1}{\mathbb{E}} \left[ #2  \right]}

\newcommand{\pr}[1]{\left(#1\right)}
\newcommand{\NP}{\mathsf{NP}}

\newcommand{\OPT}[0]{{\ensuremath{\sf{OPT}}}}
\newcommand{\CP}[0]{{\ensuremath{\sf{CP}}}}

\newcommand{\diag}{\textup{diag}}

\DeclareMathOperator{\tr}{tr}

\newcommand{\E}[1]{\ensuremath{\mathbb{E}}\text{{$\left[#1\right]$}}}

\DeclareMathAlphabet{\altmathcal}{OMS}{cmsy}{m}{n}

\newtheorem{theorem}{Theorem}[section]
\newtheorem{lemma}[theorem]{Lemma}
\newtheorem{corollary}[theorem]{Corollary}

\newtheorem{definition}[theorem]{Definition}
\newtheorem{claim}[theorem]{Claim}
\newtheorem{observation}[theorem]{Observation}
\newtheorem{remark}[theorem]{Remark}

\newenvironment{proof}{\smallskip\noindent{\bf Proof}:}{$\hfill \Box$\\}

\newenvironment{proofof}[1]{\smallskip\noindent{\bf Proof of #1}:}{$\hfill \Box$\\}

\newcommand{\Tao}[1]{\color{green} Tao's Comment: {#1} \color{black}}

\def\I{\mathcal{I}}

\def\R{\mathbb{R}}
\def\BB{\altmathcal{B}}
\def\Z{\mathbb{Z}}

\def\U{\altmathcal{U}}

\def\A{\mathcal{A}}

\def\SS{\altmathcal{S}}

\newcommand{\cut}[1]{}

\title{Proportional Volume Sampling and Approximation Algorithms for $A$-Optimal Design}

\usepackage{times}

\author{Aleksandar Nikolov\thanks{University of Toronto. Email:anikolov@cs.toronto.edu.} and Mohit Singh\thanks{Georgia Institute of Technology. Email:mohitsinghr@gmail.com.  Supported in part by
National Science Foundation grant CCF-24067E5.} and Uthaipon Tao Tantipongpipat\thanks{Georgia Institute of Technology. Email:tao@gatech.edu. Supported in part by
National Science Foundation grants CCF-24067E5 and CCF-1740776 (Georgia Institute of Technology TRIAD), and by a
Georgia Institute of Technology ARC fellowship.}}
\usepackage{varioref}
\usepackage{amstext}

\begin{document}

\maketitle

\begin{abstract}
We study the $A$-optimal design problem where we are given vectors $v_1,\ldots, v_n\in \R^d$, an integer $k\geq d$, and the goal is to select a set $S$ of $k$ vectors that minimizes the trace of $\left(\sum_{i\in  S} v_i v_i^{\top}\right)^{-1}$. Traditionally, the problem is an instance of optimal design of experiments in statistics (\cite{pukelsheim2006optimal}) where each vector corresponds to a linear measurement of an unknown vector and the goal is to pick $k$ of them that minimize the average variance of the error in the maximum likelihood estimate of the vector being measured. The problem also finds applications in sensor placement in wireless networks~(\cite{joshi2009sensor}), sparse least squares regression~(\cite{BoutsidisDM11}), feature selection for $k$-means clustering~(\cite{boutsidis2013deterministic}), and matrix approximation~(\cite{de2007subset,de2011note,avron2013faster}). In this paper, we introduce \emph{proportional volume sampling} to obtain improved approximation algorithms for $A$-optimal design.

Given a matrix, proportional volume sampling involves picking a set of columns $S$ of size $k$ with probability proportional to $\mu(S)$ times $\det(\sum_{i \in S}v_i v_i^\top)$ for some measure $\mu$. Our main result is to show the approximability of the $A$-optimal design problem can be reduced to \emph{approximate} independence properties of the measure $\mu$. We appeal to hard-core distributions as candidate distributions $\mu$ that allow us to obtain improved approximation algorithms for the $A$-optimal design. Our results include a $d$-approximation when $k=d$, an $(1+\epsilon)$-approximation when $k=\Omega\left(\frac{d}{\epsilon}+\frac{1}{\epsilon^2}\log\frac{1}{\epsilon}\right)$ and $\frac{k}{k-d+1}$-approximation when repetitions of vectors are allowed in the solution. We also consider generalization of the problem for $k\leq d$ and obtain a $k$-approximation.

We also show that the proportional volume sampling algorithm gives approximation algorithms for other optimal design objectives (such as $D$-optimal design~\cite{singh2018approximate} and generalized ratio objective~\cite{mariet2017elementary}) matching or improving previous best known results. Interestingly, we show that a similar guarantee cannot be obtained for the $E$-optimal design problem.
We also show that the $A$-optimal design problem is $\NP$-hard to approximate within a fixed constant when $k=d$.
\end{abstract}



\section{Introduction}

Given a collection of vectors, a common problem is to select a subset of size $k\leq n$ that \emph{represents} the given vectors. To quantify the representability of the chosen set, typically one considers spectral properties of certain natural matrices defined by the vectors. Such problems arise as experimental design~(\cite{fedorov1972theory,pukelsheim2006optimal}) in statistics; feature selection~(\cite{boutsidis2013deterministic}) and sensor placement problems~(\cite{joshi2009sensor}) in machine learning; matrix sparsification~(\cite{batson2012twice,spielman2011graph}) and  column subset selection~(\cite{avron2013faster}) in numerical linear algebra. In this work, we consider the optimization problem of choosing the representative subset that aims to optimize the \emph{$A$-optimality criterion} in experimental design.

Experimental design is a classical problem in statistics~(\cite{pukelsheim2006optimal}) with recent applications in machine learning~(\cite{joshi2009sensor,wang2016computationally}). Here the goal is to estimate an unknown vector $w\in \R^d$ via linear measurements of the form $y_i= v_i^\top w + \eta_i$ where $v_i$ are possible experiments and $\eta_i$ is assumed to be small i.i.d. unbiased Gaussian error introduced in the measurement. Given a set $S$ of linear measurements, the maximum likelihood estimate $\hat{w}$ of $w$ can be obtained via a least squares computation. The error vector $w-\hat{w}$ has a Gaussian distribution with mean $0$ and  covariance matrix $\left(\sum_{i\in S} v_i v_i^\top\right)^{-1}$. In the optimal experimental design problem the goal is to pick a cardinality $k$ set $S$ out of the $n$ vectors such that the measurement error is minimized. Minimality is measured according to different criteria, which quantify the ``size'' of the covariance matrix. In this paper, we study the classical $A$-optimality criterion, which aims to minimize the average variance over directions, or equivalently the trace of the covariance matrix, which is also the expectation of the squared Euclidean norm of the error vector $w-\hat{w}$.

We let $V$ denote the $d\times n$ matrix whose columns are the vectors $v_1,\ldots, v_n$ and $[n]=\{1,\ldots, n\}$. For any set $S\subseteq [n]$, we let $V_S$ denote the $d\times |S|$ submatrix of $V$ whose columns correspond to vectors indexed by $S$. Formally, in the $A$-optimal design problem our aim is to find a subset $S$ of cardinality $k$ that minimizes the trace of $(V_SV_S^\top)^{-1} = \left( \sum_{i\in S} v_i v_i^\top\right)^{-1}$. We also consider the $A$-optimal design problem with repetitions, where the chosen $S$ can be a multi-set, thus allowing a vector to chosen more than once.

Apart from experimental design, the above formulation finds application in other areas such as sensor placement in wireless networks~(\cite{joshi2009sensor}), sparse least squares regression~(\cite{BoutsidisDM11}), feature selection for $k$-means clustering~(\cite{boutsidis2013deterministic}), and matrix approximation~(\cite{avron2013faster}). For example, in matrix approximation~(\cite{de2007subset,de2011note,avron2013faster}) given a $d \times n$ matrix $V$, one aims to select a set $S$ of $k$  such that the Frobenius norm of the Moore-Penrose pseudoinverse of the selected matrix $V_S$ is minimized. It is easy to observe that this objective equals the $A$-optimality criterion for the vectors given by the columns of $V$.



\subsection{Our Contributions and Results}
\cut{
\begin{table}[h!]
\begin{center}\begin{tabular}{|c|c|c|c|}
\hline
Paper & Case \(k=d\) & Without Repetition  & With Repetition \\\hline
\cite{avron2013faster} & \(n-d+1\) & \(\frac{n-d+1}{k-d+1}\) & \(\frac{n-d+1}{k-d+1}\) \\\hline
\cite{wang2016computationally}, sampling & - & \(1+\epsilon\), for \(k\geq\Omega(d^2/\epsilon)\) & \(1+\epsilon\) for \(k\geq\Omega(d\log d/\epsilon^2)\) \textsuperscript{[1]} \\\hline
\cite{wang2016computationally}, greedy & - & \(1+\frac{d(d+1)}{2(k-d+1)}\) & \(1+\frac{d(d+1)}{2(k-d+1)}\) \\\hline
\cite{allen2017near} & - & \(1+\epsilon\), for \(k\geq\Omega(d/\epsilon^2) \textsuperscript{[2]}\) & \(1+\epsilon\), for \(k\geq\Omega(d/\epsilon^2)\) \\\hline
Our paper & \(d\) \textsuperscript{[3]} & \(1+\epsilon\), for \(k\geq\Omega(\frac{d}{\epsilon}+\frac{\log \epsilon}{\epsilon^2})\) & \(\frac{k}{k-d+1}\) \textsuperscript{[3]} \\\hline
\end{tabular}\end{center}
\caption{Summary of approximation ratio of  \(A\)-optimal results. Any result in without repetition case implies with repetition case. We list applications of results to \(k=d\), if applicable, for comparison. \textsuperscript{[1]}By outputting \(S\) of size \(O(k)\) instead of \(k\). \textsuperscript{[2]} The ratio applies to \(A,D,E,V,G\)-optimality. We proved in this paper that this ratio is tight for \(E\)-optimality. \textsuperscript{[3]}The ratio tight up to the convex relaxation \eqref{eq:CP-obj}-\eqref{eq:CP-bounds}.}
\label{table:summary-result}
\end{table}
}
\begin{table}
\begin{center}\begin{tabular}{|c|c|c|}
\hline
Problem & Our result & Previous work \\\hline
Case \(k=d\) & \(d\) \textsuperscript{[1]}  & \(n-d+1\) (\cite{avron2013faster}) \\\hline
\begin{tabular}{@{}c@{}} Asymptotic \(k>>d\) \\ without Repetitions\end{tabular}   & \(1+\epsilon\), for \(k\geq\Omega\left(\frac{d}{\epsilon}+\frac{\log 1/\epsilon}{\epsilon^2}\right)\) & \(1+\epsilon\), for \(k\geq\Omega\left(\frac{d}{\epsilon^2}\right) \) (\cite{allen2017near}) \\\hline
\begin{tabular}{@{}c@{}} Arbitrary $k$ and $d$ \\ With Repetitions\end{tabular}   & \(\frac{k}{k-d+1}\) \textsuperscript{[1]} & \(n-d+1\) (\cite{avron2013faster}) \\\hline
\begin{tabular}{@{}c@{}} Asymptotic \(k>> d\) \\ With Repetitions\end{tabular}   & \(1+\epsilon\), for \(k\geq d+\frac{d}{\epsilon}\) \textsuperscript{[1]} & \(1+\epsilon\), for \(k\geq\Omega(\frac{d}{\epsilon^2})\)  (\cite{allen2017near}) \\\hline
\end{tabular}\end{center}
\caption{Summary of approximation ratios of  \(A\)-optimal results. We list the best applicable previous work for comparison. \textsuperscript{[1]}The ratios are tight with matching integrality gap of the convex relaxation \eqref{eq:CP-obj}-\eqref{eq:CP-bounds}. 
 }
\label{table:summary-result}
\end{table}

Our main contribution is to introduce the \emph{proportional volume sampling} class of probability measures to obtain improved approximation algorithms for the $A$-optimal design problem. We obtain improved algorithms for the problem with and without repetitions in regimes where $k$ is close to $d$ as well as in the asymptotic regime where $k\geq d$. The improvement is summarized in Table \ref{table:summary-result}.  Let $\altmathcal{U}_k$ denote the collection of subsets of $[n]$ of size exactly $k$ and $\altmathcal{U}_{\leq k}$ denote the subsets of $[n]$ of size at most $k$. We will consider distributions on sets in $\U_k$ as well as $\U_{\leq k}$ and state the following definition more generally.

\begin{definition}
Let $\mu$ be probability measure on sets in $\U_k$ (or $\U_{\leq k}$).  Then the proportional volume sampling with measure $\mu$ picks a set \cut{of vectors indexed by} ${S}\in \U_k$ (or $\U_{\leq k}$) with probability proportional to $\mu(S) \det(V_SV_S^\top)$.
\end{definition}

Observe that when $\mu$ is the uniform distribution and $k\leq d$ then we obtain the standard volume sampling~(\cite{DeshpandeR10}) where one picks a set $S$ proportional to $\det(V_SV_S^\top)$, or, equivalently, to the volume of the parallelopiped spanned by the vectors indexed by $S$.   The volume sampling measure has received much attention and efficient algorithms are known for sampling from it~(\cite{DeshpandeR10,guruswami2012optimal}). More recently, efficient algorithms were obtained even when $k\geq d$~(\cite{li2017column,singh2018approximate}). We discuss the computational issues of sampling from proportional volume sampling in Lemma~\ref{lem:sample} and Section~\ref{sec:deterministic}.

Our first result shows that approximating the $A$-optimal design problem can be reduced to finding distributions on $\U_k$ (or $\U_{\leq k}$) that are \emph{approximately independent}. First, we define the exact formulation of approximate independence needed in our setting.

\begin{definition}
Given integers $d\leq k\leq n$ and a vector $x\in [0,1]^n$ such that $1^\top x=k$, we call a measure $\mu$ on sets in $\U_k$ (or $\U_{\leq k}$), $\alpha$-approximate $(d-1,d)$-wise independent with respect to \(x\) if for any subsets $T, R\subseteq [n]$ with $|T|=d-1$ and $|R|=d$, we have
$$ \frac{\mathsf{Pr}_{\altmathcal{S}\sim \mu}[T\subseteq \SS]}{\mathsf{Pr}_{\altmathcal{S}\sim \mu}[R\subseteq \SS]}\leq \alpha \frac{x^T}{x^R}$$
where \(x^L:=\prod_{i\in L}x_i\) for any $L\subseteq [n]$.
We omit ``with respect to \(x\)" when the context is clear.  \end{definition}
Observe that if the measure $\mu$ corresponds to picking each element $i$ independently with probability $x_i$, then  $\frac{\mathsf{Pr}_{\altmathcal{S}\sim \mu}[T\subseteq \SS]}{\mathsf{Pr}_{\altmathcal{S}\sim \mu}[R\subseteq \SS]}= \frac{x^T}{x^R}$. However, this distribution has support on all sets and not just sets in $\U_k$ or $\U_{\leq k}$, so it is not allowed by the definition above.

Our first result reduces the search for approximation algorithms for
$A$-optimal design to construction of approximate $(d-1,d)$-wise
independent distributions. This result generalizes the connection
between volume sampling and $A$-optimal design established
in~\cite{avron2013faster} to proportional volume sampling, which
allows us to exploit the power of the convex relaxation and get a
significantly improved approximation.
\begin{theorem}\label{thm:nearind-to-Aopt}
Given integers $d\leq k\leq n$, suppose that for any a vector $x\in [0,1]^n$ such that $1^\top x=k$ there exists a distribution $\mu$ on sets in $\U_k$ (or $\U_{\leq k}$) that is $\alpha$-approximate $(d-1,d)$-wise independent. Then the proportional volume sampling with measure $\mu$ gives an $\alpha$-approximation algorithm for the $A$-optimal design problem.
\end{theorem}

In the above theorem, we in fact only need an approximately independent distribution $\mu$ for the optimal solution $x$ of the natural convex relaxation for the problem, which is given in \eqref{eq:CP-obj}--\eqref{eq:CP-bounds}. The result also bounds the integrality gap of the convex relaxation by $\alpha$. Theorem~\ref{thm:nearind-to-Aopt} is proved in Section~\ref{sec:approx-independent}.

Theorem~\ref{thm:nearind-to-Aopt} reduces our aim to constructing distributions that have approximate $(d-1,d)$-independence. We focus our attention on the general class of \emph{hard-core distributions}.  We call $\mu$ a \emph{hard-core} distribution with parameter $\lambda\in \R^n_+$ if $\mu(S)\propto \lambda^S:= \prod_{i\in S} \lambda_i$ for each set in $\U_k$ (or $\U_{\leq k}$). Convex duality implies that hard-core distributions have the maximum entropy among all distributions which match the marginals of $\mu$~(\cite{boyd2004convex}). Observe that, while $\mu$ places non-zero probability on exponentially many sets, it is enough to specify $\mu$ succinctly by describing $\lambda$. Hard-core distributions over various structures including spanning trees~(\cite{gharan2011randomized}) or matchings~(\cite{kahn1996asymptotics,kahn2000asymptotics}) in a graph display \emph{approximate independence} and this has found use in combinatorics as well as algorithm design. Following this theme, we show that certain hard core distributions on $\U_k$ and $\U_{\le k}$  exhibit approximate $(d-1,d)$-independence when $k=d$ and in the asymptotic regime when $k>>d$.

\begin{theorem}\label{thm:k=dAndasymptotic}
Given integers $d\leq k\leq n$ and a vector $x\in [0,1]^n$ such that $1^\top x=k$, there exists a hard-core distribution $\mu$ on sets in $\U_k$ that is $d$-approximate $(d-1,d)$-wise independent when $k=d$. Moreover, for any $\epsilon >0$, if $k=\Omega\left(\frac{d}{\epsilon}+\frac{1}{\epsilon^2} \log{\frac{1}{\epsilon}}\right)$, then there is a hard-core distribution $\mu$ on $\U_{\leq k}$  that is $(1+\epsilon)$-approximate $(d-1,d)$-wise independent. Thus we obtain a $d$-approximation algorithm for the $A$-optimal design problem when $k=d$ and $(1+\epsilon)$-approximation algorithm when $k=\Omega\left(\frac{d}{\epsilon}+\frac{1}{\epsilon^2} \log{\frac{1}{\epsilon}}\right)$.
\end{theorem}

The above theorem relies on two natural hard-core distributions. In the first one, we consider the hard-core distribution with parameter $\lambda= x$ on sets in $\U_k$ and in the second we consider the hard-core distribution with parameter $\lambda=\frac{(1-\epsilon)x}{1- (1-\epsilon)x}$ (defined co-ordinate wise) on sets in $\U_{\leq k}$. We prove the theorem in Section~\ref{sec:without-rep}. 

Our techniques also apply to the $A$-optimal design problem with repetitions where we obtain an even stronger result, described below. The main idea is to introduce multiple, possibly exponentially many, copies of each vector, depending on the fractional solution, and then apply proportional volume sampling to obtain the following result.

\begin{theorem}\label{thm:repetitions}
For all \(k\geq d\) and \(0<\epsilon\leq 1\), there is a $(\frac{k}{k-d+1}+\epsilon)$-approximation algorithm for the $A$-optimal design problem with repetitions. In particular, there is a $(1+\epsilon)$-approximation when $k\geq d+ \frac{d}{\epsilon}$.
\end{theorem}

We remark that the integrality gap of the natural convex relaxation is at least $\frac{k}{k-d+1}$ (see Section~\ref{sec:integralityA}) and thus the above theorem results in an exact characterization of the integrality gap of the convex program  \eqref{eq:CP-obj}--\eqref{eq:CP-bounds}, stated in the following corollary. The proof of Theorem~\ref{thm:repetitions} appears in Section~\ref{sec:effreplace}.

\begin{corollary}
  For any integers $k\geq d$, the integrality gap of the convex program  \eqref{eq:CP-obj}--\eqref{eq:CP-bounds} for the $A$-optimal design with repetitions is exactly $\frac{k}{k-d+1}$.
\end{corollary}


We also show that $A$-optimal design is $\NP$-hard for $k=d$ and moreover, hard to approximate
within a constant factor. 

\begin{theorem}\label{thm:hardness}
  There exists a constant $c > 1$ such that the $A$-optimal design
  problem is $\NP$-hard to $c$-approximate when $k=d$.
\end{theorem}

\paragraph{The $k \leq d$ case.} The $A$-optimal design problem has a
natural extension to choosing fewer than $d$ vectors: our objective in
this case is to select a set $S \subseteq [n]$ of size $k$ so that we
minimize $\sum_{i = 1}^k\lambda_i^{-1}$, where $\lambda_1, \ldots,
\lambda_k$ are the $k$ largest eigenvalues of the matrix
$V_SV_S^\top$. While this problem no longer corresponds to
minimizing the variance in an experimental design setting, we will
abuse terminology and still call it the $A$-optimal design problem. This is
 a natural formulation of the geometric problem of picking a set of vectors which
are as ``spread out'' as possible. If $v_1, \ldots, v_n$ are the points in
a dataset, we can see an optimal solution as a
maximally diverse representative sample of the dataset. Similar
problems, but with a determinant objective, have been widely studied
in computational geometry, linear algebra, and machine learning: for
example the largest volume simplex problem, and the maximum
subdeterminant problem (see~\cite{nikolov2015randomized} for
references to prior work). \cite{ccivril2009selecting} also studied
an analogous problem with the sum in the objective replaced by a
maximum (which extends $E$-optimal design).

While our rounding extends easily to the $k \le d$ regime, coming up
with a convex relaxation becomes less trivial. We do find such a
relaxation and obtain the following result whose proof appears in Section~\ref{sec:smallk}.

\begin{theorem}\label{thm:ksmallStatementAtIntro}
There exists a \(\poly(d,n)\)-time $k$-approximation algorithm for the $A$-optimal design problem when $k \le d$.
\end{theorem}

\paragraph{General Objectives.} Experimental design problems come with many different objectives including $A$, $D$, $E$, $G$, $T$, $V$, each corresponding to a different function of the covariance matrix of the error $w-\hat{w}$. We note that any algorithm that solves \(A\)-optimal objective can solve \(T\)-optimal objective by prepossessing vectors with a linear transformation. In addition, we show that the  proportional volume sampling algorithm gives approximation algorithms for other optimal design objectives (such as $D$-optimal design~\cite{singh2018approximate} and generalized ratio objective~\cite{mariet2017elementary}) matching or improving previous best known results. We refer the reader to Section~\ref{sec:general} for details. 

\paragraph{Integrality Gap and $E$-optimality.} Given the results mentioned above, a natural question is whether all objectives for optimal design behave similarly in terms of approximation algorithms. Indeed, recent results of \cite{allen2017near,Allen-Zhu17nearoptimal} and \cite{wang2016computationally} give the  $(1+\epsilon)$-approximation algorithm in the asymptotic regime, $k\geq \Omega\left(\frac{d}{\epsilon^2}\right)$ and $k\geq \Omega\left(\frac{d^2}{\epsilon}\right)$, for many of these variants. In contrast, we show the \emph{optimal bounds} that can be obtained via the standard convex relaxation are different for different objectives.  We
show that for the $E$-optimality criterion (in which we minimize the
largest eigenvalue of the covariance matrix) getting a $(1+\epsilon)$-approximation with the
natural convex relaxation requires $k = \Omega(\frac{d}{\epsilon^2})$,
both with and without repetitions. This is in sharp contrast to results we obtain here for $A, D$-optimality and other generalized ratio objectives. Thus, different criteria behave
differently in terms of approximability. Our proof of the integrality gap (in
Section~\ref{sec:alon-boppana}) builds on a connection to spectral
graph theory and in particular on the Alon-Boppana
bound~(\cite{Alon86,Nilli91}). We prove an Alon-Boppana
style bound for the unnormalized Laplacian of not necessarily regular
graphs with a given average degree.

\paragraph{Restricted Invertibility Principle for Harmonic Mean.}
As an application of Theorem~\ref{thm:ksmallStatementAtIntro}, we
prove a restricted invertibility principle (RIP)~(\cite{bour-tza}) for
the harmonic mean of singular values. The RIP is a robust version of
the elementary fact in linear algebra that if $V$ is a $d \times n$
rank $r$ matrix, then it has an invertible submatrix $V_S$ for some
$S\subseteq [n]$ of size $r$. The RIP shows that if $V$ has stable
rank $r$, then it has a well-invertible submatrix consisting of
$\Omega(r)$ columns. Here the stable rank of $V$ is the ratio
$(\|V\|_{HS}^2/\|V\|^2)$, where $\|\cdot\|_{HS} = \sqrt{\tr(VV^\top)}$
is the Hilbert-Schmidt, or Frobenius, norm of $V$, and $\|\cdot\|$ is
the operator norm. The classical restricted invertibility
principle~(\cite{bour-tza,vershynin,bt-constructive}) shows that when
the stable rank of $V$ is $r$, then there exists a subset of its
columns $S$ of size $k = \Omega(r)$ so that the $k$-th singular value
of $V_S$ is
$\Omega\left({\|V\|_{HS}}/{\sqrt{m}}\right)$.~\cite{nikolov2015randomized}
showed there exists a submatrix $V_S$ of $k$ columns of $V$ so that
the geometric mean its top $k$ singular values is on the same order,
even when $k$ {equals} the stable rank.  We show an analogous result
for the harmonic mean when $k$ is slightly less than $r$. While this
is implied by the classical restricted invertibility principle, the
dependence on parameters is better in our result for the harmonic
mean. For example, when $k = (1-\epsilon)r$, the harmonic mean of
squared singular values of $V_S$ can be made at least
$\Omega\left(\epsilon {\|V\|^2_{HS}}/{{m}}\right)$, while the tight
restricted invertibility principle of Spielman and
Srivastava~(\cite{spielman2011graph}) would only give $\epsilon^2$ in
the place of $\epsilon$. This restricted invertibility principle can
also be derived from the results of~\cite{Naor17}, but their
arguments, unlike ours, do not give an efficient algorithm to compute
the submatrix $V_S$. See Section~\ref{sect:rip} for the precise
formulation of our restricted invertibility principle.

\paragraph{Computational Issues.}
While it is not clear whether sampling from proportional volume sampling is possible under general assumptions (for example given a sampling oracle for $\mu$), we obtain an efficient sampling algorithm when $\mu$ is a hard-core distribution. 
\begin{lemma}\label{lem:sample}
There exists a \(\poly(d,n)\)-time algorithm that, given a matrix $d\times n$ matrix  $V$, integer $k\leq n$, and a hard-core distribution $\mu$ on sets in $\U_k$ (or $\U_{\leq k}$) with parameter $\lambda$, efficiently samples a set from the proportional volume measure defined by $\mu$.

\end{lemma}
When $k\leq d$ and $\mu$ is a hard-core distribution, the proportional volume sampling can be implemented by the standard volume sampling after scaling the vectors appropriately. When $k>d$, such a method does not suffice and we appeal to properties of hard-core distributions to obtain the result.  We also present an efficient implementation of Theorem~\ref{thm:repetitions} which runs in time polynomial in \(\log(1/\epsilon)\). This requires more work since the basic description of the algorithm involves implementing proportional volume sampling on an exponentially-sized ground set. This is done in Section~\ref{sec:effreplace}.

We also outline efficient deterministic implementation of algorithms in Theorem~\ref{thm:k=dAndasymptotic} and \ref{thm:repetitions} in Section~\ref{sec:deterministic} and \ref{Sec:deterministicAlgwithRep}.

\subsection{Related Work}

Experimental design is the problem of maximizing information obtained from selecting subsets of experiments to perform, which is equivalent to minimizing the covariance matrix $\left(\sum_{i\in S} v_i v_i^\top\right)^{-1}$. We focus on \(A\)-optimality, one of the criteria that has been studied intensely. We restrict our attention to approximation algorithms for these problems and refer the reader to~\cite{pukelsheim2006optimal} for a broad survey on experimental design.

\cite{avron2013faster} studied the $A$- and $E$-optimal design problems and analyzed various combinatorial algorithms and algorithms based on volume sampling, and achieved approximation ratio $\frac{n-d+1}{k-d+1}$. 
\cite{wang2016computationally} found connections between optimal design and matrix sparsification, and used these connections to obtain a $(1+\epsilon)$-approximation when $k\geq \frac{d^2}{\epsilon}$, and also approximation algorithms under certain technical assumptions.
More recently, \cite{allen2017near,Allen-Zhu17nearoptimal} obtained a  $(1+\epsilon)$-approximation when $k=\Omega\left(\frac{d}{\epsilon^2}\right)$ both with and without repetitions.
We remark that their result also applies to other criteria such as $E$ and $D$-optimality that aim to maximize the minimum eigenvalue, and the geometric mean of the eigenvalues of $\sum_{i\in S} v_i v_i^\top$, respectively. More generally, their result applies to any objective function that satisfies certain regularity criteria. 



Improved bounds for $D$-optimality were obtained by \cite{singh2018approximate} who give an ${e}$-approximation for all $k$ and $d$, and $(1+\epsilon)$-approximation algorithm when $k=\Omega(\frac{d}{\epsilon}+\frac{1}{\epsilon^2}\log\frac{1}{\epsilon})$, with a weaker condition of $k\geq\frac{2d}{\epsilon}$ if repetitions are allowed. The $D$-optimality criterion when $k\leq d$ has also been extensively studied. It captures maximum a-posteriori inference in constrained determinantal point process models~(\cite{kulesza2012determinantal}), and also the maximum volume simplex problem. \cite{nikolov2015randomized}, improving on a long line of work, gave a $e$-approximation. The problem has also been studied under more general matroid constraints rather than cardinality constraints~(\cite{nikolov2016maximizing,AnariG17,StraszakV17}).

\cite{ccivril2009selecting} also studied several related problems in the $k \le d$ regime, including $D$- and $E$-optimality. We are not aware of any prior work on $A$-optimality in this regime.

The criterion of $E$-optimality, whose objective is to maximize the minimum eigenvalue of $\sum_{i\in S} v_i v_i^\top$, is closely related to the problem of matrix sparsification~(\cite{batson2012twice,spielman2011graph}) but incomparable. 
 In matrix sparsification, we are allowed to weigh the selected vectors, but need to bound both the largest and the smallest eigenvalue of the matrix we output.

 The restricted invertibility principle was first proved in the work of~\cite{bour-tza}, and was later strengthened by~\cite{vershynin}, \cite{bt-constructive}, and \cite{Naor17}. Spielman and Srivastava gave a deterministic algorithm to find the well-invertible submatrix whose existence is guaranteed by the theorem. 
 Besides its numerous applications in geometry (see~\cite{vershynin} and \cite{Youssef14}), the principle has also found applications to differential privacy~(\cite{NTZ}), and to approximation algorithms for discrepancy~(\cite{apx-disc}).

Volume sampling where a set $S$ is sampled with probability proportional to $\det(V_SV_S^\top)$ has been studied extensively and efficient algorithms were given by \cite{DeshpandeR10} and improved by~\cite{guruswami2012optimal}. The probability distribution is also called a \textit{determinantal point process} (DPP) and finds many applications in machine learning~(\cite{kulesza2012determinantal}). Recently, fast algorithms for volume sampling have been considered in \cite{derezinski2017subsampling,derezinski2017unbiased}.

While $\mathsf{NP}$-hardness is known for the $D$- and $E$-optimality criteria~(\cite{ccivril2009selecting}), to the best of our knowledge no $\mathsf{NP}$-hardness for $A$-optimality was known prior to our work. Proving such a hardness result was stated as an open problem in~\cite{avron2013faster}.



%


\section{Approximation via Near Independent Distributions}\label{sec:approx-independent}

In this section, we prove Theorem~\ref{thm:nearind-to-Aopt} and give
an $\alpha$-approximation algorithm for the $A$-optimal design problem
given an $\alpha$-approximate $(d-1,d)$-independent distribution $\mu$.

We first consider the convex relaxation for the problem given below
for the settings without and with repetitions. This relaxation is
classical, and already appears in, e.g.~\cite{chernoff1952measure}. It is easy to see that the objective $\tr \left(\sum_{i = 1}^n x_i v_i  v_i^\top\right)^{-1}$ is convex~(\cite{boyd2004convex}, section 7.5). For this section, we focus on the case when repetitions are not allowed.

\begin{center}
\begin{tabular}[h]{|p{6.5cm}|p{6.5cm}|} \hline 
\vspace{0.5em}\qquad\qquad \qquad \textbf{With Repetitions}{\begin{align*}
  \qquad &\min \tr \left(\sum_{i = 1}^n x_i v_i  v_i^\top\right)^{-1}\\
  & \text{s.t.} \qquad \sum_{i = 1}^n{x_i} = k\\
  &\qquad0 \le x_i   \  \ \ \ \forall i \in [n]
\end{align*}} & \vspace{0.5em}\qquad\qquad\quad \textbf{Without Repetitions}{\begin{align}
  \qquad &\min \tr \left(\sum_{i = 1}^n x_i v_i  v_i^\top\right)^{-1}\label{eq:CP-obj}\\
  & \text{s.t.} \qquad \sum_{i = 1}^n{x_i} = k\\
  &\qquad0 \le x_i \le 1  \  \ \ \ \forall i \in [n]\label{eq:CP-bounds}
\end{align}} \\ \hline
\end{tabular}
\end{center}

\cut{
}
Let us denote the optimal value of
\eqref{eq:CP-obj}--\eqref{eq:CP-bounds} by \CP. By plugging in the indicator vector of an optimal integral solution for $x$, we see that $\CP \le \OPT$, where $\OPT$ denotes the value of the optimal solution.

\subsection{Approximately Independent Distributions}
Let us use the notation $x^S = \prod_{i \in S}x_i$, $V_S$ a matrix of column vectors $v_i\in \R^d$ for $i\in S$, and $V_S(x)$ a matrix of column vectors $\sqrt{x_i} v_i\in \R^d$ for $i\in S$. Let $e_k(x_1, \ldots, x_n)$ be the degree $k$
elementary symmetric polynomial in the variables $x_1, \ldots,
x_n$, i.e.
$e_k(x_1, \ldots, x_n) = \sum_{S \in \U_k} x^S$.
By convention, $e_0(x) = 1$ for any $x$. For any positive
semidefinite $n\times n$ matrix $M$, we define $E_k(M)$ to be
$e_k(\lambda_1, \ldots, \lambda_n)$, where $\lambda(M)=(\lambda_1, \ldots,
\lambda_n)$ is the vector of eigenvalues of $M$. Notice that $E_1(M) =
\tr(M)$ and $E_n(M) = \det(M)$.


To prove Theorem~\ref{thm:nearind-to-Aopt}, we give
Algorithm \ref{alg:prop-vol-samping} which is a general framework to sample $S$ to
solve the \(A\)-optimal design problem.

\begin{algorithm}
\caption{The proportional volume sampling algorithm }\label{alg:prop-vol-samping}
\begin{algorithmic}[1]
\State Given an input $V=[v_1,\ldots,v_n]$ where $v_i\in \R^d$, $k$ a positive integer, and measure  $\mu$ on sets in $\U_k$ (or $\U_{\leq k}$).
\State Solve convex relaxation \CP \ to get a fractional solution $x\in \R_+^n$ with $\sum_{i=1}^n x_i = k$.
\State Sample set $\SS$ (from $\altmathcal{U}_{\leq k}$ or $\altmathcal{U}_k$) where $\Prob{}{\SS=S}\propto \mu(S)\det(V_SV_S^\top)$ for any $S\in \U_k$ (or $\U_{\leq k}$). \Comment{$\mu(S)$ may be defined using the solution \(x\)}
\State Output $\SS$ (If $|\SS| < k$, add $k - |\SS|$ arbitrary vectors to $\SS$ first).
\end{algorithmic}
\end{algorithm}


We first prove the following lemma which is needed for proving Theorem~\ref{thm:nearind-to-Aopt}.\begin{lemma} \label{lemma:LinearlityofDeterminant}
\label{lemma:linearofDetInX}Let $T\subseteq [n]$ be of size no more than $d$. Then \[\det(V_T(x)^\top V_T(x)) = x^T \det(V_T^\top V_T) \]
\end{lemma}
\begin{proof}
The statement is true  by multilinearity of the determinant and the exact formula for $V_T(x)^\top V_T(x)$ as follows. The matrix  $V_T(x)^\top V_T(x)$ has \((i,j)\) entry
\[\left(V_T(x)^\top V_T(x)\right)_{i,j} = \sqrt{x_i}v_i \cdot \sqrt{x_j}v_j = \sqrt{x_ix_j}v_i\cdot v_j\]
for each pair $i,j\in [|T|]$. By the multilinearity of the
determinant, we can take the factor $\sqrt{x_i}$ out from each row $i$
of  $V_T(x)^\top V_T(x)$ and the factor $\sqrt{x_j}$ out from each column $j$ of $V_T(x)^\top V_T(x)$. This gives
\[\det(V_T(x)^\top V_T(x)) = \prod_{i \in [|T|]}\sqrt{x_i} \prod_{j \in [|T|]}\sqrt{x_j} \det(V_T^\top V_T) = x^T\det(V_T^\top V_T)  \]
\end{proof}
We also need the following identity, which is well-known and extends
the Cauchy-Binet formula  for the determinant to the functions $E_k$.
\begin{equation}
  \label{eq:cauchy-binet}
  E_k(VV^\top) = E_k(V^\top V)= \sum_{S \in \U_k}{\det(V_S^\top V_S)}.
\end{equation}
The identity \eqref{eq:cauchy-binet} appeared in  \cite{mariet2017elementary} and, specifically for \(k=d-1,\)  as Lemma 3.8 in \cite{avron2013faster}. \cut{Let us sketch a proof of \eqref{eq:cauchy-binet} for general \(k\) for
completeness. Since the nonzero eigenvalues of $VV^\top$ and $V^\top
V$ are the same, we have $E_k(VV^\top) = E_k(V^\top V)$. To show the
second equality, we can consider the univariate polynomial $p(z) =
\det(zI + V^\top V)$. On one hand, expanding $p(z)$ using the Leibniz
formula for the determinant shows that the coefficient of $z^{n-k}$ in
$p(z)$ is equal to the right hand side of \eqref{eq:cauchy-binet}. On
the other hand, if $\lambda = \lambda(V^\top V)$ is the vector of
eigenvalues of $V^\top V$, then we have $p(z) = \prod_{i = 1}^n{(z +
  \lambda_i)}$, and it is easy to see that the coefficient of $z^{n-k}$
in $p(z)$ is $e_k(\lambda) = E_k(V^\top V)$, as required.
} Now we are ready to prove Theorem~\ref{thm:nearind-to-Aopt}.

\begin{proofof}{Theorem~\ref{thm:nearind-to-Aopt}}
Let \(\mu'\) denote the sampling distribution over \(\U\), where $\U=\altmathcal{U}_k$ or $\U_{\leq k}$, with probability of sampling \(S\in\U\) proportional to $\mu(S)\det(V_SV_S^\top)$. Because $\tr \left( \sum_{i\in [n]} x_i v_i v_i^\top \right)^{-1} = \CP\leq\OPT$, it is enough to show that
\begin{equation}
\Expectation{\SS\sim \mu'}{\tr \left( \sum_{i\in \SS} v_i v_i^\top \right)^{-1}}
\leq \alpha \tr \left( \sum_{i\in [n]} x_i v_i v_i^\top
\right)^{-1}.  \label{eq:wise-ind-give-approx-alpha-bound}
\end{equation}  
Note that in case $|\SS| < k$, algorithm $\mathcal{A}$ adds $k -
|\SS|$ arbitrary vector to $\SS$, which can only decrease the
objective value of the solution.

First, a simple but important observation (\cite{avron2013faster}): for any $d\times d$ matrix $M$ of rank $d$,
we have
\begin{equation}\label{eq:keyobsv}
\tr M^{-1} = \sum_{i = 1}^d \frac{1}{\lambda_i(M)} =
\frac{e_{d-1}(\lambda(M))}{e_d(\lambda(M))} = \frac{E_{d-1}(M)}{\det M}.
\end{equation}
 Therefore, we have
 \begin{align*}
   \Expectation{\SS\sim \mu'}{\tr \left( \sum_{i\in \SS} v_i v_i^\top \right)^{-1}} &= \sum_{S\in \U} \Prob{\mu'}{\SS=S}\tr \left( V_S V_S^\top \right)^{-1} \\
   &=  \sum_{S\in \U} \frac{\mu(S) \det \left( V_S V_S^\top \right) } {\sum_{S'\in \U} \mu(S') \det(V_{S'} V_{S'}^\top)}\frac{E_{d-1}(V_SV_S^\top)}{\det\left(V_S V_S^\top\right)} \\
   &= \frac{\sum_{S\in \U} \mu(S) E_{d-1}(V_SV_S^\top)} {\sum_{S\in   \U} \mu(S) \det(V_S V_S^\top)}.
 \end{align*}
 We can now apply the Cauchy-Binet formula \eqref{eq:cauchy-binet} for  $E_{d-1}$, $E_d = \det$, and the matrix $V_SV_S^\top$ to the
 numerator and denominator on the right hand side, and we get
 \begin{align*}
   \Expectation{\SS\sim \mu'}{\tr \left( \sum_{i\in \SS} v_i v_i^\top \right)^{-1}}
   &= \frac{\sum_{S\in \U} \sum_{|T|=d-1,T\subseteq S} \mu(S) \det(V_T^\top V_T)} {\sum_{S\in \U} \mu(S) \sum_{|R|=d,R \subseteq S} \det(V_R^\top V_R) } \\
    &= \frac{\sum_{|T|=d-1,T\subseteq[n]} \det \left( V_T^\top V_T \right) \sum_{S\in\U,S \supseteq T} \mu(S)  }     { \sum_{|R|=d,R\subseteq[n]} \det \left( V_R^\top V_R \right) \sum_{S\in\U,S \supseteq R} \mu(S) } \\
    &=  \frac{\sum_{|T|=d-1,T\subseteq[n]} \det \left( V_T^\top V_T \right) \Prob{\mu}{\SS\supseteq T}  }     { \sum_{|R|=d,R\subseteq[n]} \det \left( V_R^\top V_R \right)\Prob{\mu}{\SS\supseteq R}  }
 \end{align*}
  where we change the order of summation at the second to last
  equality. Next, we apply \eqref{eq:keyobsv} and the Cauchy-Binet formula
  \eqref{eq:cauchy-binet} in a similar way to the matrix $V(x) V(x)^\top$:
\begin{align*}
        \tr \left( V(x) V(x)^\top \right)^{-1}  = \frac{E_{d-1}(V(x)V(x)^\top)}{\det(V(x)V(x)^\top)}
        &= \frac{\sum_{|T|=d-1,T\subseteq[n]} \det(V_T(x)^\top V_T(x))}{\sum_{|R|=d,R\subseteq[n]}\det(V_R(x)^\top V_R(x))} \\
    &= \frac{\sum_{|T|=d-1,T\subseteq[n]} \det \left( V_T^\top V_T \right) x^T}     { \sum_{|R|=d,R\subseteq[n]} \det \left( V_R^\top V_R  \right) x^R }
 \end{align*}
 where we use the fact that $\det(V_R(x)^\top V_R(x)) = x^R \det(V_R^\top  V_R)$ and $\det(V_T(x)^\top V_T(x)) = x^T \det(V_T^\top V_T)    $ in the last equality by Lemma \ref{lemma:linearofDetInX}.

Hence, the inequality \eqref{eq:wise-ind-give-approx-alpha-bound} which we want to show is equivalent to
 \begin{align}
   \frac{\sum_{|T|=d-1,T\subseteq[n]} \det \left( V_T^\top V_T \right) \Prob{\mu}{\SS\supseteq T}  }     { \sum_{|R|=d,R\subseteq[n]} \det \left( V_R ^\top V_R \right)\Prob{\mu}{\SS\supseteq R}  } \leq\alpha
   \frac{\sum_{|T|=d-1,T\subseteq [n]} \det \left( V_T^\top V_T \right) x^T}     { \sum_{|R|=d,R\subseteq[n]} \det \left( V_R^\top V_R  \right) x^R }\label{eq:CauchyBinetOnObjOfRelaxation}
 \end{align}
  which is equivalent to
\begin{eqnarray}
 &&\sum_{|T|=d-1,|R|=d} \det \left( V_T^\top V_T \right) \det \left( V_R^\top V_R  \right)\cdot x^R\cdot \Prob{\mu}{\SS\supseteq T}\nonumber \\
 && \leq\alpha\sum_{|T|=d-1,|R|=d} \det \left( V_T^\top V_T \right) \det \left( V_R^\top V_R  \right)\cdot x^T\cdot\Prob{\mu}{\SS\supseteq R}. \label{eq:thm-near-ind-expand-sum-last}
\end{eqnarray}
By the assumption that $ \frac{\Prob{\mu}{\SS\supseteq T}}{\Prob{\mu}{\SS\supseteq R}} \leq \alpha\frac{x^{T}}{x^{R}} $ for each subset $T,R\subseteq[n]$ with $|T|=d-1$ and $|R|=d$, 
\begin{equation}\det \left( V_T^\top V_T \right) \det \left( V_R^\top V_R  \right)\cdot  x^R\cdot \Prob{\mu}{\SS\supseteq T}\leq\alpha \det \left( V_T^\top V_T \right) \det \left( V_R^\top V_R  \right)\cdot x^T\cdot \Prob{\mu}{\SS\supseteq R} \label{eq:thm-near-ind-expand-last}
\end{equation}
Summing \eqref{eq:thm-near-ind-expand-last} over all $T,R$ proves \eqref{eq:thm-near-ind-expand-sum-last}.
\end{proofof}

\section{Approximating Optimal Design without Repetitions}\label{sec:without-rep}

In this section, we prove Theorem~\ref{thm:k=dAndasymptotic} by
constructing $\alpha$-approximate $(d-1,d)$-independent distributions
for appropriate values of $\alpha$. We first consider the case when
$k=d$ and then the asymptotic case when
$k=\Omega\left(\frac{d}{\epsilon}+\frac{1}{\epsilon^2}\log
  \frac{1}{\epsilon}\right)$. We also remark that the argument for
$k=d$ can be generalized for all $k\leq d$, and we discuss this generalization in Section~\ref{sec:smallk}.
\subsection{$d$-approximation for $k=d$}
We prove the following lemma which, together with
Theorem~\ref{thm:nearind-to-Aopt}, implies the $d$-approximation for
$A$-optimal design when $k=d$.
\begin{lemma} \label{lem:k=dresult}
Let $k=d$. The hard-core distribution $\mu$ on $\U_k$ with parameter $x$ is $d$-approximate $(d-1,d)$-independent.
\end{lemma}
\begin{proof}
Observe that for any $S\in \U_k$, we have $\mu(S)= \frac{ x^S}{Z}$
where $Z=\sum_{S'\in \U_k} x^{S'}$ is the normalization factor. For any $T\subseteq [n]$ such that $|T|=d-1$, we have
$$\Prob{\SS\sim \mu}{\SS \supseteq T}= \sum_{S\in \U_k: S\supseteq T} \frac{x^S}{Z}= \frac{x^T}{Z}\cdot \left(\sum_{i\in [n]\setminus T} x_i\right) \leq d \frac{x^T}{Z}.$$
where we use  $k=d$ and $\sum_{i\in [n]\setminus T} x_i \leq k=d$.
For any $R\subseteq [n]$ such that $|R|=d$, we have
\[\Prob{\SS\sim \mu}{\SS \supseteq R}= \sum_{S\in \U_k: S\supseteq R} \frac{x^S}{Z}= \frac{x^R}{Z}.\]
Thus for any $T, R\subseteq [n]$ such that $|T|=d-1$ and $|R|=d$, we have
\[\frac{\Prob{\SS\sim \mu}{\SS\supseteq T}}{\Prob{\SS\sim \mu}{\SS\supseteq R}} \leq d \frac{x^T}{x^R}.\]
\end{proof}

\subsection{$(1+\epsilon)$-approximation}
Now, we show that there is a hard-core distribution $\mu$ on $\U_{\leq k}$ that is $(1+\epsilon)$-approximate $(d-1,d)$-independent when $k=\Omega\left(\frac{d}{\epsilon}+\frac{1}{\epsilon^2}\log \frac{1}{\epsilon}\right)$.

\begin{lemma} \label{lm:asymp-ind}
  Fix some $0 < \epsilon \le 2$, and let $k =
  \Omega\left(\frac{d}{\epsilon} +    \frac{\log(1/\epsilon)}{\epsilon^2}\right)$.
  The hard-core distribution $\mu$ on $\U_{\le k}$ with parameter
  $\lambda$, defined by
  \[
  \lambda_i = \frac{x_i}{1 + \frac{\epsilon}{4} - x_i},
  \]
  is $(1+\epsilon)$-approximate $(d-1, d)$-wise independent.
\end{lemma}
\begin{proof}
  For simplicity of notation, let us denote $\beta =
  1+\frac{\epsilon}{4}$, and $\xi_i = \frac{x_i}{\beta}$. Observe that
  the probability mass under $\mu$ of any set $S$ of size at most $k$
  is proportional to $\left(\prod_{i \in S}{\xi_i}\right)
  \left(\prod_{i \not \in S}{(1-\xi_i)}\right)$. Thus,
  $\mu$ is equivalent to the following distribution: sample a set $\BB
  \subseteq [n]$ by including every $i \in [n]$ in $\BB$ independently
  with probability $\xi_i$; then we have $\mu(S) = \Pr[\BB = S
  \mid |\BB| \le k]$ for every $S$ of size at most $k$.  Let us fix
  for the rest of the proof arbitrary sets $T,R\subseteq [n]$ of size
  $d-1$ and $d$, respectively. By the
  observation above, for $\SS$ sampled according to $\mu$, and $\BB$
  as above, we have
  \[
  \frac{\Pr[\SS \supseteq T]}{\Pr[\SS \supseteq R]}
  =
  \frac{\Pr[\BB \supseteq T \text{ and } |\BB| \le k]}{\Pr[\BB
    \supseteq R \text{ and } |\BB| \le k]}
  \le
  \frac{\Pr[\BB \supseteq T]}{\Pr[\BB \supseteq R \text{ and } |\BB| \le k]}
  \]
  We have $\Pr[\BB \supseteq T] = \xi^T = \frac{x^T}{\beta^{d-1}}$. To
  simplify the probability in the denominator, let us introduce, for  each $i \in [n]$, the indicator random variable $Y_i$, defined to
  be $1$ if $i \in \BB$ and $0$ otherwise. By the choice of $\BB$, the
  $Y_i$'s are independent Bernoulli random variables with mean
  $\xi_i$, respectively. We can write
  \begin{align*}
  \Pr[\BB \supseteq R \text{ and } |\BB| \le k]
  &=
  \Pr\Biggl[\forall i \in R: Y_i = 1 \text{ and } \sum_{i \not \in
    R}{Y_i} \le k-d\Biggr]\\
  &=   \Pr[\forall i \in R: Y_i = 1] \Pr\Biggl[\sum_{i \not \in  R}{Y_i} \le k-d\Biggr],
  \end{align*}
  where the last equality follows by the independence of the
  $Y_i$. The first probability on the right hand side is just
  $\xi^R = \frac{x^R}{\beta^d}$, and plugging into the inequality above, we get
  \begin{equation}
    \label{eq:coins-ineq}
     \frac{\Pr[\SS \supseteq T]}{\Pr[\SS \supseteq R]}
     \le
     \beta \frac{x^T}{x^R \Pr[\sum_{i \not \in  R}{Y_i} \le k-d]}.
  \end{equation}
 We claim that \[ \Pr[\sum_{i \not
    \in R}{Y_i} \le k-d] \ge 1 - \frac{\epsilon}{4}\]
as long as $k=\Omega\left(\frac{d}{\epsilon}+ \frac{1}{\epsilon^2}\log \frac{1}{\epsilon}\right)$.
  The proof follows from standard concentration of measure arguments.  Let $Y = \sum_{i \not \in
    R}{Y_i}$, and observe that $\E Y = \frac{1}{\beta}(k - x(R))$,
  where $x(R)$ is shorthand for $\sum_{i \in R}{x_i}$. By Chernoff's
  bound,
  \begin{equation}
    \label{eq:chernoff}
    \Pr[Y > k-d] < e^{-\frac{\delta^2}{3\beta} (k-x(R))}
  \end{equation}
  where
  \[
  \delta = \frac{\beta(k-d)}{k-x(R)} - 1
  = \frac{(\beta - 1)k + x(R) - \beta d}{k - x(R)}.
  \]
  The exponent on the right hand side of \eqref{eq:chernoff}
  simplifies to
  \[
  \frac{\delta^2(k-x(R))}{3\beta}
  =
  \frac{((\beta - 1)k + x(R) - \beta d)^2}{3\beta(k - x(R))}
  \ge
  \frac{((\beta - 1)k - \beta d)^2}{3\beta k}.
  \]
  For the bound $\Pr[Y > k-d] \le \frac{\epsilon}{4}$, it suffices to have
  \[
  (\beta - 1)k - \beta d \ge \sqrt{3\beta\log(4/\epsilon) k}.
  \]
  Assuming that $k \ge \frac{C\log(4/\epsilon)}{\epsilon^2}$ for a
  sufficiently big constant $C$, the right hand side is at most
  $\frac{\epsilon k}{8}$. So, as long as $k \ge \frac{\beta d}{\beta - 1 -
    \frac{\epsilon}{8}}$, the inequality is satisfied and $\Pr[Y > k-d] <
  \frac{\epsilon}{4}$, as we claimed.

 The proof of the lemma now follows since for any $|T|=d-1$ and $|R|=d$, we have
 \begin{equation}
 \label{eq:coins-ineq2}
     \frac{\Pr[\SS \supseteq T]}{\Pr[\SS \supseteq R]}
     \le
     \beta \frac{x^T}{x^R \Pr[\sum_{i \not \in  R}{Y_i} \le k-d]}
     \le
     \frac{1 + \frac{\epsilon}{4}}{1 - \frac{\epsilon}{4}} \frac{x^T}{x^R},
  \end{equation}

 and $\frac{1 + \frac{\epsilon}{4}}{1 - \frac{\epsilon}{4}} \le 1 + \epsilon$.
\end{proof}

The $(1+\epsilon)$-approximation for large enough $k$ in
Theorem~\ref{thm:k=dAndasymptotic} now follows directly from
Lemma~\ref{lm:asymp-ind} and Theorem~\ref{thm:nearind-to-Aopt}.

\section{Approximately Optimal Design with Repetitions}
In this section, we consider the $A$-optimal design without the bound
$x_i \leq 1$ and prove Theorem~\ref{thm:repetitions}. That is, we
allow the sample set $S$ to be a multi-set. We obtain a tight bound on
the integrality gap in this case. Interestingly, we reduce the problem
to a special case of $A$-optimal design without repetitions that
allows us to obtained an improved approximation.

We first describe a sampling Algorithm \ref{alg:duplicateSampleSamex_i} that achieves a
$\frac{k(1+\epsilon)}{k-d+1}$-approximation for any \(\epsilon>0\).
In the 
algorithm, we introduce \(\poly(n,1/\epsilon)\) number of copies of each vector to ensure that
the fractional solution assigns equal fractional value for each copy
of each vector. Then we use the proportional volume sampling where the
measure distribution $\mu$ is defined on sets of the new larger ground set $U$
over copies of the original input vectors. The distribution \(\mu\) is
just the uniform distribution over subsets of size $k$ of $U$, and we
are effectively using traditional volume sampling over $U$.  Notice,
however, that the distribution over multisets of the original set of
vectors is different.
The proportional volume sampling used in the algorithm can be implemented in the same way as the one used for  without repetition setting, as described in~Section
\ref{Sec:EfficientRandomizedProportionalVolume}, which runs in \(\poly(n,d,k,1/\epsilon)\) time.

In
Section~\ref{sec:effreplace}, we describe a new implementation of proportional volume sampling procedure  which improves the running time to \(\poly(n,d,k,\log(1/\epsilon))\). The new algorithm is still efficient even when \(U\) has exponential size by exploiting the facts that \(\mu\) is uniform and that \(U\) has only at most \(n\) distinct vectors. 
\begin{algorithm}
\caption{Approximation Algorithm for $A$-optimal design with repetitions}\label{alg:duplicateSampleSamex_i}
\begin{algorithmic}[1]
\State Given $x\in \R_+^n$ with $\sum_{i=1}^n x_i = k$, \(\epsilon > 0\), and vectors $v_1,\ldots,v_n$.
\State Let $q=\frac{2n}{\epsilon k}$. Set \(x_i':=\frac{k-n/q}{k}x_i\) for each \(i\), and round each \(x_i'\) up to a multiple of \(1/q\).  \label{algStep:qCanBe-qk}
\State If $\sum_{i=1}^n x_i' < k$, add \(1/q\) to any \(x_i'\) until $\sum_{i=1}^n x_i' = k$.
\State Create $qx_i'$ copies of vector $v_i$ for each $i\in[n]$. Denote  $W$  the set of size $\sum_{i=1}^nqx_i'=qk$ of all those copies of vectors. Denote \(U\)  the new index set of \(W\)of size $qk$.
\Comment{This implies that we can assume that our new fractional solution $y_i=1/q$ is equal over all $i\in U$} \label{step:copyingStepAlgRepetition}
\State Sample a subset $\SS$ of $U$ of size $k$  where  $\Pr[\SS=S]\propto \det(W_{S}W_S^\top)$ for each  $S\subseteq U$ of size $k$. \label{step:samplingStepAlgRepetition}
\label{algstep:Sampling}
\State Set $X_i=\sum_{w\in W_{\SS}} \mathbbm{1}(w \text{ is a copy of } v_i)$ for all $i\in[n]$ \Comment{Get an integral solution $X$ by counting numbers of copies of \(v_{i}\) in \(\SS\).}
\State Output $X$.
\end{algorithmic}
\end{algorithm}
\begin{lemma} \label{lem:factorOfRoundingButInefficient}
Algorithm~\ref{alg:duplicateSampleSamex_i}, when given as input $x \in \R_+^n$ s.t.~$\sum_{i=1}^n x_i = k$, $1\geq \epsilon > 0,$ and $v_1,\ldots,v_n$, outputs a random $X\in \Z_+^n$ with $\sum_{i=1}^n X_i=k$ such that
\[ \Expectation{}{\tr\left( \sum_{i=1}^n X_i v_i v_i^\top \right)^{-1}} \leq \frac{k(1+\epsilon)}{k-d+1} \tr\left( \sum_{i=1}^n x_iv_i v_i^\top \right)^{-1} \]
\end{lemma}
\begin{proof}
Define \(x_i', y,W,U,\SS ,X\)  as in the algorithm. We will show that
\begin{align*}
\Expectation{}{\tr\left( \sum_{i=1}^n X_i v_i v_i^\top \right)^{-1}} \leq \frac{k}{k-d+1} \tr\left( \sum_{i=1}^n x_i'v_i v_i^\top \right)^{-1}\leq \frac{k(1+\epsilon)}{k-d+1} \tr\left( \sum_{i=1}^n x_iv_i v_i^\top \right)^{-1}
\end{align*}
The second inequality is by observing that the scaling \(x_i':=\frac{k-n/q}{k}x_i\) multiplies the objective \(\tr\left( \sum_{i=1}^n x_iv_i v_i^\top \right)^{-1}\) by a factor of \(\left(\frac{k-n/q}{k}\right)^{-1}=(1-\epsilon/2)^{-1}\leq 1+\epsilon\), and that rounding \(x_i\) up and adding \(1/q\) to any \(x_i\) can only decrease the objective.

To show the first inequality, we first translate the two key quantities $\tr\left( \sum_{i=1}^n x_i 'v_i v_i^\top \right)^{-1}$ and $\tr\left( \sum_{i=1}^n X_i v_i v_i^\top \right)^{-1}$  from the with-repetition setting over \(V\) and \([n]\) to the without-repetition setting over \(W\) and \(U\).
First, $\tr\left( \sum_{i=1}^n x_i 'v_i v_i^\top \right)^{-1} = \tr\left( \sum_{i\in U} y_i w_i {w_i}^\top \right)^{-1}$, where $y_i = \frac{1}{q}$ are all equal over all \(i\in U\), and $w_{i}$ is the copied vector in \(W\) at index $i\in U$. Second,
$\tr\left( \sum_{i=1}^n X_i v_i v_i^\top \right)^{-1}=\tr\left( \sum_{i\in \SS\subseteq U  } w_i {w_i}^\top \right)^{-1} $.
 
Let \(\mu'\) be the  distribution over  subsets $S$ of $U$ of size $k$ defined by \(\mu'(S)\propto\det(W_S W_S^\top)\). It is, therefore, sufficient to show that the sampling distribution \(\mu'\)  satisfies
\begin{equation}
\Expectation{\altmathcal{S}\sim \mu'}{\tr\left( \sum_{i\in \altmathcal{S}\subseteq U} w_i {w_i}^\top \right)^{-1}} \leq \frac{k}{k-d+1} \tr\left( \sum_{i\in U} y_i w_i {w_i}^\top \right)^{-1}  \label{eq:with-rep-near-ind-distribution}
\end{equation}
Observe that \(\mu'\) is the same as sampling a set $S\subseteq U$ of size $k$ with probability proportional to $\mu(S)\det(W_S W_S^\top)$ where \(\mu\) is uniform.
Hence, by Theorem \ref{thm:nearind-to-Aopt}, it is enough to show that for all $T,R\subseteq U$ with $|T|=d-1,|R|=d$,
\begin{align}
\frac{\Prob{\mu}{\altmathcal{S}\supseteq T}}{\Prob{\mu}{\altmathcal{S}\supseteq R}} \leq \left(\frac{k}{k-d+1}\right)\frac{y^{T}}{y^{R}} \label{eq:with-rep-mu-requirement}
\end{align}
With $\mu$ being uniform and  $y_i$ being all equal to $1/q$,
 the calculation is straightforward:
\begin{equation} \frac{\Prob{\mu}{\altmathcal{S}\supseteq T}}{\Prob{\mu}{\altmathcal{S}\supseteq R}} = \frac{\binom{qk-d+1}{k-d+1}/\binom{qk}{k}}{\binom{qk-d}{k-d}/\binom{qk}{k}} =\frac{qk-d+1}{k-d+1} \ \text{ and } \ \frac{y^{T}}{y^{R}}= \frac{1}{y_i} = q \label{eq:with-rep-direct-calcuation} \end{equation}
Therefore, \eqref{eq:with-rep-mu-requirement} holds because
\[ \frac{\Prob{\mu}{\altmathcal{S}\supseteq T}}{\Prob{\mu}{\altmathcal{S}\supseteq R}} \cdot\left( \frac{y^{T}}{y^{R}} \right)^{-1}= \frac{qk-d+1}{k-d+1}\cdot \frac{1}{q} \leq \frac{qk}{k-d+1}\cdot \frac{1}{q} = \frac{k}{k-d+1}, \]
\end{proof}

\begin{remark}
The approximation ratio for A-optimality with repetitions for $k\geq d$ is tight, since it matches the integrality gap lower bound stated in Theorem \ref{thm:IntegralityGapLowerBound}.
\end{remark}

\bibliographystyle{plainnat}
\bibliography{Reference}
\appendix
\section{Generalizations}\label{sec:generalization}

In this section we show that our arguments extend to the regime $k \le
d$ and give a $k$-approximation (without repetitions), which matches
the integrality gap of our convex relaxation. We also derive a
restricted invertibility principle for the harmonic mean of
eigenvalues.
\subsection{$k$-Approximation Algorithm for $k\leq d$}\label{sec:smallk}

Recall that our aim is to select a set $S \subseteq [n]$ of size $k\leq d$ that
minimizes $\sum_{i = 1}^k\lambda_i^{-1}$, where $\lambda_1, \ldots,
\lambda_k$ are the $k$ largest eigenvalues of the matrix
$V_SV_S^\top$.
We need to reformulate our convex relaxation since when $k < d$, the
inverse of $M(S) = \sum_{i\in S}{v_i v_i^\top}$ for $|S| = k$ is
no longer well-defined. We write a new convex program:
\begin{align}
  &\min \frac{E_{k-1}\left(\sum_{i = 1}^n x_i v_i  v_i^\top\right)}{E_{k}\left(\sum_{i = 1}^n x_i v_i  v_i^\top\right)}\label{eq:CPksmall-obj}\\
  &\text{s.t.} \notag\\
  &\sum_{i = 1}^n{x_i} = k\\
  &0 \le x_i \le 1 \ \ \ \forall i \in [n]\label{eq:CPksmall-bounds}
\end{align}
Once again we denote the optimal value of
\eqref{eq:CPksmall-obj}--\eqref{eq:CPksmall-bounds} by \CP. While the
proof that this relaxes the original problem is easy, the convexity is
non-trivial. Fortunately, ratios of symmetric polynomials are known to
be convex.

\begin{lemma}\label{lm:ksmall-convexity}
  The optimization problem
  \eqref{eq:CPksmall-obj}--\eqref{eq:CPksmall-bounds} is a convex
  relaxation of the $A$-optimal design problem when $k \le d$.
\end{lemma}
\begin{proof}
  To prove convexity, we first note that the function $f(M) =
  \frac{E_{k}(M)}{E_{k-1}(M)}$ is concave on positive semidefinite
  matrices $M$ of rank at least $k$. This was proved by \cite[Theorem 4]{BullenM61} for positive definite $M$, and
  can be extended to $M$ of rank at least $k$ by a limiting
  argument. Alternatively, we can use the theorem of~\cite{MarcusLopes57} that the function $g(\lambda) =
  \frac{e_{k}(\lambda)}{e_{k-1}(\lambda)}$ is concave on vectors
  $\lambda \in \R^d$ with non-negative entries and at least $k$
  positive entries.  Because $g$ is  symmetric under permutations of
  its arguments and concave, and $f(M) = g(\lambda(M))$, where
  $\lambda(M)$ is the vector of eigenvalues of $M$, by a classical
  result of~\cite{Davis57} the function $f$ is concave on
  positive semidefinite matrices of rank at least $k$.

  Notice that the objective \eqref{eq:CPksmall-obj} equals
  $\frac{1}{f(M(x))}$ for the linear matrix-valued function $M(x) =
  \sum_{i = 1}^n x_i v_i v_i^\top$. Therefore, to prove that
  \eqref{eq:CPksmall-obj} is convex in $x$ for non-negative $x$, it
  suffices to show that $\frac{1}{f(M)}$ is convex in $M$ for positive
  semidefinite $M$. Since the function $\frac1z$ is convex and
  monotone decreasing over positive reals $z$, and $f$ is concave
  and non-negative over positive semidefinite matrices of rank at
  least $k$, we have that $\frac{1}{f(M)}$ is convex in $M$, as
  desired. Then   \eqref{eq:CPksmall-obj}--\eqref{eq:CPksmall-bounds}
  is an optimization problem with a convex objective and affine
  constraints, so we have a convex optimization problem.

  Let \OPT\ be the optimal value of the $A$-optimal design problem, and
  let $S$ be an optimal solution. We need to show that $\CP \le
  \OPT$. To this end, let $x$ be the indicator vector of $S$,
  i.e.~$x_i = 1$ if and only if $i \in S$, and $x_i = 0$
  otherwise. Then,
  \[
  \CP \le \frac{E_{k-1}(M(S))}{E_k(M(S))}
  = \frac{\sum_{i = 1}^k \prod_{j \neq i}{\lambda_j(M(S))}}{\prod_i  \lambda_i(M(S))}
  = \sum_{i = 1}^k\frac{1}{\lambda_i(M(S))} = \OPT.
  \]
  Above, $\lambda_1(M(S)), \ldots, \lambda_k(M(S))$ are, again, the
  nonzero eigenvalues of $M(S) = \sum_{i \in S}{v_i v_i^\top}$.
\end{proof}

We shall use the natural analog of proportional volume sampling: given
a measure $\mu$ on subsets of size $k$, we sample a set $S$
with probability proportional to $\mu(S)E_k(M(S))$. In fact, we will
only take $\mu(S)$ proportional to $x^S$, so this reduces to sampling
$S$ with probability proportional to $E_k(\sum_{i \in S}{x_i v_i
  v_i^\top})$, which is the standard volume sampling with vectors
scaled by $\sqrt{x_i}$, and can be implemented efficiently using,
e.g.~the algorithm of~\cite{DeshpandeR10}.

The following version of Theorem~\ref{thm:nearind-to-Aopt} still
holds with this modified proportional volume sampling. The proof is
exactly the same, except for mechanically replacing every instance of
determinant by $E_k$, of $E_{d-1}$ by $E_{k-1}$, and in general of $d$ by
$k$.

\begin{theorem}\label{thm:smallk-nearind}
  Given integers $k \le d \leq n$ and a vector $x\in [0,1]^n$ such that
  $1^\top x=k$, suppose there exists a measure $\mu$ on $\U_k$ that is
  $\alpha$-approximate $(k-1,k)$-wise independent. Then for $x$ the
  optimal solution of
  \eqref{eq:CPksmall-obj}--\eqref{eq:CPksmall-bounds},
  proportional volume sampling with measure $\mu$ gives a polynomial
  time $\alpha$-approximation algorithm for the $A$-optimal design
  problem.
\end{theorem}

We can now give the main approximation guarantee we have for $k \le
d$.
\begin{theorem}\label{thm:ksmall}
  For any $k \le d$, proportional volume sampling with the hard-core measure
  $\mu$ on $\U_k$ with parameter $x$ equal to the optimal solution of
  \eqref{eq:CPksmall-obj}--\eqref{eq:CPksmall-bounds} gives a
  $k$-approximation to the $A$-optimal design problem.
\end{theorem}
\begin{proof}
  In view of Theorem~\ref{thm:smallk-nearind}, we only need to show
  that $\mu$ is $k$-approximate $(k-1, k)$-wise independent. This is a
  straightforward calculation: for $\SS \sim \mu$, and any $T\subseteq
  [n]$ of size $k-1$ and $R \subseteq [n]$ of size $k$,
  \[
  \frac{\Pr[\SS \supseteq T]}{\Pr[\SS \supseteq R]}
  =
  \frac{x^T\sum_{i \not \in T}{x_i}}{x^R} \le k \frac{x^T}{x^R}.
  \]
  This completes the proof.
\end{proof}

The algorithm can be derandomized using the method of conditional
expectations analogously to the case of $k=d$ that we will show in Theorem~\ref{thm:deterministic}.

The
$k$-approximation also matches the integrality gap of
\eqref{eq:CPksmall-obj}--\eqref{eq:CPksmall-bounds}. Indeed, we can
take a $k$-dimensional integrality gap instance $v_1, \ldots, v_n$,
and embed it in $\R^d$ for any $d > k$ by padding each vector with
$0$'s. On such an instance, the convex program
\eqref{eq:CPksmall-obj}--\eqref{eq:CPksmall-bounds} is equivalent to the convex program \eqref{eq:CP-obj}--\eqref{eq:CP-bounds}. Thus the integrality gap that we will show in Theorem~\ref{thm:IntegralityGapLowerBound} implies an integrality gap of $k$ for all $d\geq k$.

\subsection{Restricted Invertibility Principle for Harmonic Mean}
\label{sect:rip}

Next we state and prove our restricted invertibility principle for harmonic mean
in a  general form.

\begin{theorem}\label{thm:rip-hm}
  Let $v_1, \ldots, v_n \in \R^d$, and $c_1, \ldots, c_n \in \R_+$,
  and define $M = \sum_{i = 1}^n{c_i v_i v_i^\top}$. For any $k \le
  r = \frac{\tr(M)}{\|M\|}$, there exists a subset $S \subseteq [n]$ of
  size $k$ such that the $k$ largest eigenvalues $\lambda_1, \ldots,
  \lambda_k$ of the matrix $\sum_{i \in S}{v_i v_i^\top}$ satisfy
  \[
  \left(\frac1k \sum_{i = 1}^k\frac{1}{\lambda_i}\right)^{-1}
  \ge
  \frac{r - k + 1}{r} \cdot \frac{\tr(M)}{\sum_{i = 1}^nc_i}.
  \]
  Moreover, such a set $S$ can be computed in deterministic
  polynomial time.
\end{theorem}
\begin{proof}
  Without loss of generality we can assume that $\sum_{i = 1}^n{c_i} =
  k$. Then, by Theorem~\ref{thm:ksmall}, proportional volume sampling
  with the hard-core measure $\mu$ on $\U_k$ with parameter $c = (c_1, \ldots, c_n)$
  gives a random set $\SS$ of size $k$ such that
  \[
  \E{\sum_{i = 1}^k\frac{1}{\lambda_i(M(\SS))}} \le k
  \frac{E_{k-1}(M)}{E_k(M)},
  \]
  where $\lambda_i(M(\SS))$ is the $i$-th largest eigenvalues of
  $M(\SS) = \sum_{i \in S}{v_i v_i^\top}$. Therefore, there
  exists a set $S$ of size $k$ such that
  \[
  \left(\frac{1}{k}\sum_{i = 1}^k{\frac{1}{\lambda_i(M(S))}}\right)^{-1}
  \ge   \frac{E_k(M)}{E_{k-1}(M)} = \frac{e_k(\lambda(M))}{e_{k-1}(\lambda(M))},
  \]
  where $\lambda(M)$ is the vector of eigenvalues of $M$. Moreover,
  such a set can be found in deterministic polynomial time by
  Theorem~\ref{thm:deterministic}. In the rest
  of the proof we compare the right hand side above with $\tr(M)$.

  Recall that a vector $x \in \R_+^d$ is majorized by a vector $y \in
  \R_+^d$, written $x \prec y$, if $\sum_{j = 1}^i{x_{(j)}} \le
  \sum_{j = 1}^i{y_{(j)}}$ holds for all $i \in [n]$, and $\sum_{i =
    1}^n x_i = \sum_{i = 1}^n{y_i}$. Here $x_{(j)}$ denotes the $j$-th
  largest coordinate of $x$, and similarly for $y_{(j)}$. Recall
  further that a function $f:\R_+^d \to \R$ is Schur-concave if $x
  \prec y$ implies $f(x) \ge f(y)$. The function
  $\frac{e_k(x)}{e_{k-1}(x)}$ was shown to be Schur concave by \cite{guruswami2012optimal}; alternatively, it
  is symmetric under permutations of $x$ and concave, as shown
  in~\cite{MarcusLopes57} (and mentioned above), which immediately
  implies that it is Schur-concave. We define a vector $x$ which
  majorizes $\lambda(M)$ by setting $x_i = \frac{1}{r}\sum_{i =
    1}^d{\lambda_i(M)}$ for $i \in [r]$, and $x_i = 0$ for $i > r$ (we
  assume here that $\lambda_1(M) \ge \ldots \ge \lambda_d(M)$). By
  Schur concavity,
  \[
  \frac{e_k(\lambda(M))}{e_{k-1}(\lambda(M))} \le
  \frac{e_k(x)}{e_{k-1}(x)} = \frac{r-k+1}{rk}\sum_{i = 1}^d\lambda_i(M).
  \]
  Since $\sum_{i = 1}^d\lambda_i(M)= \tr(M)$, and we assumed that
  $\sum_{i = 1}^nc_i = k$, this completes the proof of the theorem.
\end{proof}

We note that Theorem~\ref{thm:rip-hm} also follows from Lemma~18 and
and equation (12) of~\cite{Naor17}. However, their proof of their
Lemma~18 does not yield an efficient algorithm to compute the set $S$,
as it relies on a volume maximization argument. 

\subsection{The Generalized Ratio Objective}\label{sec:general}
In \(A\)-optimal design, given \(V=[v_1 \ldots\ v_n]\in \R^{d\times n}\), we  state the objective as minimizing $$ \tr\left(\sum_{i\in S} v_i v_i^\top\right)^{-1}=\frac{E_{d-1}(V_S V_S^\top)}{E_d(V_S V_S^\top)}.$$ over subsets \(S\subseteq [n]\) of size \(k\). In this section, for any given pair of integers \(0\leq l'< l \leq d\), we consider the following \textit{generalized ratio problem}: 
\begin{equation}
\min_{S\subseteq [n],|S|=k} \left(\frac{E_{l'}(V_S V_S^\top)}{E_l(V_S V_S^\top)}\right)^{\frac{1}{l-l'}} \label{eq:objective-gen}
\end{equation}

The above problem naturally interpolates between $A$-optimality and $D$-optimality. This follows since for $l=d$ and $l'=0$, the objective reduces to

\begin{equation}
\min_{S\subseteq [n],|S|=k} \left(\frac{1}{\det(V_S V_S^\top)}\right)^{\frac{1}{d}}.\label{eq:objective-gen-d}
\end{equation}

A closely related generalization between \(A\)- and \(D\)-criteria was considered in \cite{mariet2017elementary}. Indeed, their generalization corresponds to the case when $l=d$ and $l'$ takes any value from $0$ and $d-1$.

In this section, we show that our results extend to solving generalized ratio problem. 
We begin by describing a convex program for the generalized ratio problem. We then generalize the proportional volume sampling algorithm to \emph{proportional $l$-volume sampling}. Following the same plan as in the proof of $A$-optimality, we then reduce the approximation guarantee to near-independence properties of certain distribution. Here again, we appeal to the same product measure and obtain identical bounds, summarized in Table \ref{table:summary-result-gen}, on the performance of the algorithm. The efficient implementations of approximation algorithms for generalized ratio problem are described in Section \ref{sec:eff-implementation-gen}.
\begin{table}
\begin{center}\begin{tabular}{|c|c|c|c|}
\hline
Problem & \begin{tabular}{c}
\textit{A}-optimal \\
(\(l'=d-1,l=d\)) \\
\end{tabular}\cut{ \( \min\limits_{|S|=k} \tr\left(V_S V_S^\top\right)^{-1}\)}  & \(\min\limits_{|S|=k} \left(\frac{E_{l'}(V_S V_S^\top)}{E_l(V_S V_S^\top)}\right)^{\frac{1}{l-l'}}\) & \begin{tabular}{c}
\textit{D}-optimal \\
(\(l'=0,l=d\)) \\
\end{tabular} \\\hline
Case \(k=d\) & \(d\)  & \(l \cdot \left[(l-l')!\right]^{-\frac{1}{l-l'}} \leq \frac{el}{l-l'}\) & \(e\)\\\hline
\begin{tabular}{@{}c@{}} Asymptotic \(k>>d\) \\ without Repetitions\end{tabular}   & \begin{tabular}{@{}c@{}} \(1+\epsilon\), for\\ \(k\geq\Omega\left(\frac{d}{\epsilon}+\frac{\log 1/\epsilon}{\epsilon^2}\right)\) \end{tabular}&  \begin{tabular}{@{}c@{}} \(1+\epsilon\), for\\ \(k\geq\Omega\left(\frac{l}{\epsilon}+\frac{\log 1/\epsilon}{\epsilon^2}\right)\) \end{tabular}&  \begin{tabular}{@{}c@{}} \(1+\epsilon\),  for\\  \(k\geq\Omega\left(\frac{d}{\epsilon}+\frac{\log 1/\epsilon}{\epsilon^2}\right)\) \end{tabular} \\\hline
\begin{tabular}{@{}c@{}} Arbitrary $k$ and $d$ \\ With Repetitions\end{tabular}   & \(\frac{k}{k-d+1}\) & \(\frac{k}{k-l+1}\)  \textsuperscript{} & \(\frac{k}{k-d+1}\) \\\hline
\begin{tabular}{@{}c@{}} Asymptotic \(k>> d\) \\ With Repetitions\end{tabular}   & \begin{tabular}{@{}c@{}} \(1+\epsilon\), for \\ \(k\geq d+\frac{d}{\epsilon}\) \end{tabular}  & \begin{tabular}{@{}c@{}} \(1+\epsilon\), for \\ \(k\geq l+\frac{l}{\epsilon}\) \end{tabular} & \begin{tabular}{@{}c@{}} \(1+\epsilon\), for \\ \(k\geq d+\frac{d}{\epsilon}\) \end{tabular} \\\hline
\end{tabular}\end{center}
\caption{Summary of approximation ratio obtained by our work on  generalized ratio problem. 
 }
\label{table:summary-result-gen}
\end{table}




\subsubsection{Convex Relaxation} \label{sec:con-lex-gen}
\cut{\cite{mariet2017elementary} show that \(\left(E_{l'}(M)\right)^{\frac{1}{l'}}\) is  convex on the set of positive definite matrices \(M\). }
As in solving \(A\)-optimality, we may define  relaxations for with and without repetitions as \eqref{eq:CP-obj-gen}-\eqref{eq:CP-bound-gen}.
\begin{center}
\begin{tabular}[h]{|p{6.5cm}|p{7.5cm}|} \hline 
\vspace{0.5em}\qquad\qquad \qquad \textbf{With Repetitions}{\begin{align*}
  \qquad &\min\pr{\frac{E_{l'}\pr{V(x) V(x)^\top   }}{E_{l}\pr{V(x) V(x)^\top }}}^{\frac{1}{l-l'}}\\
  & \text{s.t.} \qquad \sum_{i = 1}^n{x_i} = k\\
  &\qquad0 \le x_i   \  \ \ \ \forall i \in [n]
\end{align*}} & \vspace{0.5em}\qquad\qquad\quad \textbf{Without Repetitions}{\begin{align}
  \qquad &\min \pr{\frac{E_{l'}\pr{V(x) V(x)^\top   }}{E_{l}\pr{V(x) V(x)^\top }}}^{\frac{1}{l-l'}}\label{eq:CP-obj-gen}\\
  & \text{s.t.} \qquad \sum_{i = 1}^n{x_i} = k\\
  &\qquad0 \le x_i \le 1  \  \ \ \ \forall i \in [n]\label{eq:CP-bound-gen}
\end{align}} \\ \hline
\end{tabular}
\end{center}

We now show that \(\pr{\frac{E_{l'}\pr{V(x) V(x)^\top   }}{E_{l}\pr{V(x) V(x)^\top }}}^{\frac{1}{l-l'}}\) is convex in \(x\). \cut{We generalize Lemma \ref{lm:ksmall-convexity}, which applies only to \textit{A}-optimality criteria, to generalized ratio problem.}

\begin{lemma}
Let \(d\) be a positive integer. For any given pair \(0\leq l' < l \leq d\), the function \begin{align}
f_{l',l}(M)=\left(\frac{E_{l'}(M)}{E_l(M)}\right)^{\frac{1}{l-l'}} \label{eq:obj-gen-convex-scaled}
\end{align}
is convex over \(d\times d\) positive semidefinite matrix \(M\).
\end{lemma}
\begin{proof}
By Theorem 3 in \cite{BullenM61}, \((f_{l',l}(M))^{-1}=\left(\frac{E_{l}(M)}{E_{l'}(M)}\right)^{\frac{1}{l-l'}}\) is concave on positive semidefinite matrices \(M\) for each \(0\leq l' < l \leq d\). The  function $\frac1z$ is convex and
  monotone decreasing over the positive reals $z$, and this, together with
  the concavity of $(f_{l',l}(M))^{-1}$ and that \((f_{l',l}(M))^{-1}>0\), implies that $f_{l',l}(M)$ is convex in
  $M$.
\end{proof}

\subsubsection{Approximation via $(l',l)$-Wise Independent Distribution}

Let \(0\leq l' < l \leq d\) and \(\U\in\{\U_k,\U_{\leq k}\}\). We first show connection of approximation guarantees on objectives \(\left(\frac{E_{l'}(V_S V_S^\top)}{E_l(V_S V_S^\top)}\right)^{\frac{1}{l-l'}}\) and \(\frac{E_{l'}(V_S V_S^\top)}{E_l(V_S V_S^\top)}\). Suppose we already solve the convex relaxation of generalized ratio problem \eqref{eq:CP-obj-gen}-\eqref{eq:CP-bound-gen} and get a fractional solution \(x\in\R^n\).  Suppose that a randomized algorithm \(\A\), upon receiving input \(V\in \R^{d\times n}\) and \(x\in \R^n\), outputs \(S\in \U\) such that
\begin{equation}
\Expectation{S\sim\A}{\frac{E_{l'}(V_S V_S^\top)}{E_l(V_S V_S^\top)}} \leq \alpha'\frac{E_{l'}(V(x)V(x)^\top)}{E_l(V(x)V(x)^\top)} \label{eq:objective-gen-alpha'-approx}
\end{equation}
for some constant \(\alpha'>0\). By the convexity of the function \(f(z)=z^{l-l'}\) over positive reals \(z\), we have
\begin{equation}
\Expectation{}{\frac{E_{l'}(M)}{E_l(M)}}\geq \Expectation{}{\left(\frac{E_{l'}(M)}{E_l(M)}\right)^{\frac{1}{l-l'}}}^{l-l'}
\label{eq:z^(l-l')-convex}
\end{equation}
for any semi-positive definite matrix \(M\). Combining \eqref{eq:objective-gen-alpha'-approx} and \eqref{eq:z^(l-l')-convex} gives
\begin{equation}
\Expectation{S\sim\A}{\left(\frac{E_{l'}(V_S V_S^\top)}{E_l(V_S V_S^\top)}\right)^{\frac{1}{l-l'}}} \leq \alpha\left(\frac{E_{l'}(V(x)V(x)^\top)}{E_l(V(x)V(x)^\top)}\right)^{\frac{1}{l-l'}}
\label{eq:objective-gen-alpha-approx}
\end{equation}
where \(\alpha=(\alpha')^{\frac{1}{l-l'}}\). Therefore, it is sufficient for an algorithm to satisfy \eqref{eq:objective-gen-alpha'-approx} and give a bound on \(\alpha'\) in order to solve the generalized ratio problem up to factor \(\alpha\).

To show \eqref{eq:objective-gen-alpha'-approx}, we first define the proportional \(l\)-volume sampling and $\alpha$-approximate $(l',l)$-wise independent distribution.\begin{definition}
Let $\mu$ be probability measure on sets in $\U_k$ (or $\U_{\leq k}$).  Then the proportional \(l\)-volume sampling with measure $\mu$ picks a set of vectors indexed by ${S}\in \U_k$ (or $\U_{\leq k}$) with probability proportional to $\mu(S) E_l(V_SV_S^\top)$.
\end{definition}

\begin{definition}
Given integers $d,k,n$, a pair of integers \(0\leq l' \leq l \leq d\), and a vector $x\in [0,1]^n$ such that $1^\top x=k$, we call a measure $\mu$ on sets in $\U_k$ (or $\U_{\leq k}$), $\alpha$-approximate $(l',l)$-wise independent with respect to \(x\) if for any subsets $T', T\subseteq [n]$ with $|T'|=l'$ and $|T|=l$, we have
$$ \frac{\mathsf{Pr}_{\altmathcal{S}\sim \mu}[T'\subseteq \SS]}{\mathsf{Pr}_{\altmathcal{S}\sim \mu}[T\subseteq \SS]}\leq \alpha^{l-l'} \cdot \frac{x^{T'}}{x^T}$$
where \(x^L:=\prod_{i\in L}x_i\) for any $L\subseteq [n]$.
We omit ``with respect to \(x\)" when the context is clear.\end{definition}
 The following theorem reduces the approximation guarantee  in \eqref{eq:objective-gen-alpha'-approx} to $\alpha$-approximate $(l',l)$-wise independence properties of a certain distribution \(\mu\) by utilizing  proportional \(l\)-volume sampling.
\begin{theorem}\label{thm:nearind-to-Aopt-Gen}
Given integers $d,k,n$, \(V=[v_1\ldots v_n]\in\R^{d\times n}\), and a vector $x\in [0,1]^n$ such that $1^\top x=k$, suppose there exists a distribution $\mu$ on sets in $\U_k$ (or $\U_{\leq k}$) and is $\alpha$-approximate $(l',l)$-wise independent for some \(0 \leq l'  < l \leq d\). Then the proportional \(l\)-volume sampling with measure $\mu$ gives an $\alpha$-approximation algorithm for minimizing \(\left(\frac{E_{l'}(V_S V_S^\top)}{E_l(V_S V_S^\top)}\right)^{\frac{1}{l-l'}}\) over subsets \(S\subseteq[n]\) of size \(k\).

\end{theorem}
\begin{proof}
Let \(\mu'\) denote the sampling distribution over \(\U\), where $\U=\altmathcal{U}_k$ or $\U_{\leq k}$, with probability of sampling \(S\in\U\) proportional to $\mu(S)E_l(V_SV_S^\top)$. We mechanically replace \(T,R,d-1,d,\text{and }\det\) in the proof of Theorem \ref{thm:nearind-to-Aopt} with \(T',T,l',l,\text{and }E_l\) to obtain
\[ \Expectation{\SS\sim \mu'}{\tr \left( \sum_{i\in \SS} v_i v_i^\top \right)^{-1}}
\leq \alpha^{l-l'} \tr \left( \sum_{i\in [n]} x_i v_i v_i^\top
\right)^{-1}.  \]
We finish the proof by observing that \eqref{eq:objective-gen-alpha'-approx} implies \eqref{eq:objective-gen-alpha-approx}, as discussed earlier.
\end{proof}

The following subsections generalize algorithms and proofs for with and without repetitions. The algorithm for generalized ratio problem can be summarized in Algorithm \ref{alg:generalSampleNearInd}. Note that efficient implementation of the sampling is described in Section \ref{sec:eff-implementation-gen}.
\begin{algorithm}
\caption{Generalized ratio approximation algorithm}\label{alg:generalSampleNearInd}
\begin{algorithmic}[1]
\State Given an input $V=[v_1,\ldots,v_n]$ where $v_i\in \R^d$, $k$ a positive integer, and a pair of integers \(0\leq l'< l \leq d\).
\State Solve the convex relaxation \(x=\text{argmin}_{x\in J^n:1^\top x = k}\pr{\frac{E_{l'}\pr{V(x) V(x)^\top   }}{E_{l}\pr{V(x) V(x)^\top }}}^{\frac{1}{l-l'}}\) where \(J=[0,1]\) if without repetitions or \(\R^+\) if with repetitions.
\If{\(k=l\)}
\State Sample \(\mu'(S)\propto x^S E_l\pr{V_S V_S^\top }\) for each \(S \in \U_{ k}\)
\ElsIf{without repetition setting and \(k \geq \Omega\pr{\frac{d}{\epsilon}+\frac{\log(1/\epsilon)}{\epsilon^2}}\)}  \State Sample \(\mu'(S)\propto \lambda^S E_l\pr{V_S V_S^\top }\) for each \(S \in \U_{\leq k}\) where \(\lambda_i:=\frac{x_i}{1+\epsilon/4-x_i}\)
\ElsIf{with repetition setting}  \State Run Algorithm \ref{alg:duplicateSampleSamex_i}, except modifying the sampling step to sample a subset $\SS$ of $U$ of size $k$  with  $\Pr[\SS=S]\propto E_l(W_{S}W_S^\top)$.
\EndIf
\State Output $\SS$ (If $|\SS| < k$, add $k - |\SS|$ arbitrary vectors to $\SS$ first).
\end{algorithmic}
\end{algorithm}

\subsubsection{Approximation Guarantee for Generalized Ratio Problem without Repetitions}
We prove the following theorem which generalizes Lemmas \ref{lem:k=dresult} and \ref{lm:asymp-ind}. The  $\alpha$-approximate $(l',l)$-wise independence property, together with
Theorem~\ref{thm:nearind-to-Aopt-Gen}, implies an approximation guarantee for generalized ratio problem without repetitions for \(k=l\) and asymptotically for \(k=\Omega\left(\frac{l}{\epsilon}+\frac{1}{\epsilon^2} \log{\frac{1}{\epsilon}}\right)\).
\begin{theorem}\label{lem:k=dAndasymptotic-Gen}
Given integers $d,k,n,$ a pair of integers \(0\leq l' \leq l \leq d\), and a vector $x\in [0,1]^n$ such that $1^\top x=k$, the hard-core distribution $\mu$ on sets in $\U_k$ with parameter \(x\) is $\alpha$-approximate $(l',l)$-wise independent when $k=l$ for \begin{align}
\alpha=l \cdot \left[(l-l')!\right]^{-\frac{1}{l-l'}}\leq\frac{el}{l-l'} \label{eq:gen-k=d-alpha-value}
\end{align}

Moreover, for any $0<\epsilon\leq 2$ when $k=\Omega\left(\frac{l}{\epsilon}+\frac{1}{\epsilon^2} \log{\frac{1}{\epsilon}}\right)$, the hard-core distribution $\mu$ on $\U_{\le k}$ with parameter
  $\lambda$, defined by
  \[
  \lambda_i = \frac{x_i}{1 + \frac{\epsilon}{4} - x_i},
  \]
  is $(1+\epsilon)$-approximate $(l', l)$-wise independent.

Thus  for minimizing the generalized ratio problem  \(\left(\frac{E_{l'}(V_S V_S^\top)}{E_l(V_S V_S^\top)}\right)^{\frac{1}{l-l'}}\) over subsets \(S\subseteq [n]\) of size \(k\), we obtain

\begin{itemize}
\item $(\frac{el}{l-l'})$-approximation algorithm when $k=l$, and
\item $(1+\epsilon)$-approximation algorithm when $k=\Omega\left(\frac{l}{\epsilon}+\frac{1}{\epsilon^2} \log{\frac{1}{\epsilon}}\right)$.
\end{itemize}
\end{theorem}

\begin{proof}
We first prove the result for \(k=l\). For all \(T',T\subseteq [n]\) such that \(|T'|=l',|T|=l,\)\begin{align*}
\frac{\Prob{\SS\sim \mu}{S \supseteq T'}}{\Prob{\SS\sim \mu}{S \supseteq T}}=\frac{\sum_{|S|=k,S\supseteq T'}x^S}{\sum_{|S|=k,S\supseteq T}x^S}=\frac{x^{T'}\sum_{L\in\binom{[n]\setminus T'}{k-l'}}x^L}{x^T}\leq \frac{x^{T'}\sum_{L\in\binom{[n] }{k-l'}}x^L}{x^T}
\end{align*}
 We now use Maclaurin's inequality (\cite{lin1993some}) to bound the quantity on the right-hand side
\begin{align}
\sum_{L\in\binom{[n] }{k-l'}}x^L=e_{l-l'}(x)\leq \binom{n}{l-l'}\left(e_1(x)/n\right)^{l-l'}\leq\frac{n^{l-l'}}{(l-l')!}\left(l/n\right)^{l-l'}=\frac{l^{l-l'}}{(l-l')!}
\end{align}
Therefore,
\begin{equation}
\frac{\Prob{\SS\sim \mu}{S \supseteq T'}}{\Prob{\SS\sim \mu}{S \supseteq T}} \leq\frac{l^{l-l'}}{(l-l')!} \frac{x^{T'}}{x^T}
\end{equation}which proves the $(l',l)$-wise independent property of \(\mu\) with required approximation ratio from \eqref{eq:gen-k=d-alpha-value}.

We now prove the result for \( k=\Omega\left(\frac{l}{\epsilon}+\frac{1}{\epsilon^2} \log{\frac{1}{\epsilon}}\right) \). The proof follows similarly from Lemma \ref{lm:asymp-ind} by replacing \(T,R\) with \(T',T\) of sizes \(l',l\) instead of sizes \(d-1,d\). In particular, the equation \eqref{eq:coins-ineq} becomes
  \begin{equation}
    \label{eq:gen-coins-ineq}
     \frac{\Pr[\SS \supseteq T']}{\Pr[\SS \supseteq T]}
     \le
     \left(1+\frac{\epsilon}{4}\right)^{l-l'} \frac{x^{T'}}{x^T \Pr[\sum_{i \not \in  T}{Y_i} \le k-l]}.
  \end{equation}
and the Chernoff's bound \eqref{eq:chernoff} still holds by mechanically replacing \(d,R\) with \(l,T\) respectively. The resulting approximation ratio \(\alpha\) satisfies $$\alpha^{l-l'} = \frac{\left(1+\frac{\epsilon}{4}\right)^{l-l'}}{1-\frac{\epsilon}{4}}\leq (1+\epsilon)^{l-l'}$$
where the inequality holds because \(\epsilon\leq 2\).
\end{proof}
\cut{
\begin{remark}
There is a \(\alpha\)-approximation algorithm  for minimizing \ref{eq:objective-gen} for each arbitrary $k\geq d$, where \begin{align}
\alpha=\frac{k^{l-l'}}{(l-l')!}
\end{align}\label{rem:k-approxFork>d-Gen} \end{remark}
\begin{proof}
Scale the fractional solution from the convex relaxation by \(d/k\) so that the new solution satisfies \(\sum_{i\in[n]} x_i=d\), and apply the \(l\)-proportional volume sampling on the case \(k=d\).
\end{proof}}

\subsubsection{Approximation Guarantee for Generalized Ratio Problem with Repetitions}
We now consider the generalized ratio problem \textit{with repetitions}. The following statement is a generalization of Lemma  \ref{lem:factorOfRoundingButInefficient}.
\begin{theorem}\label{thm:repetitions-Gen}
Given \(V=[v_1 \ v_2 \ldots v_n]\) where \(v_i\in\R^d\), a pair of integers \(0\leq l' \leq l \leq d\), an integer $k\geq l$, and \(1\geq\epsilon>0\), there is an \(\alpha\)-approximation algorithm for minimizing   \(\left(\frac{E_{l'}(V_S V_S^\top)}{E_l(V_S V_S^\top)}\right)^{\frac{1}{l-l'}}\) over subsets \(S\subseteq [n]\) of size \(k\)  with repetitions for \begin{equation}
\alpha\leq\frac{k(1+\epsilon)}{k-l+1}
\end{equation}
\end{theorem}
\begin{proof}
We use the algorithm similar to Algorithm~\ref{alg:duplicateSampleSamex_i} except that in step \eqref{step:samplingStepAlgRepetition}, we  sample \(\SS\subseteq U\) of size \(k\) where   $\Pr[\SS=S]\propto E_l(W_{S}W_S^\top)$ in place of   $\Pr[\SS=S]\propto E_l(W_{S}W_S^\top)$. The analysis follows on the same line as in Lemma \ref{lem:factorOfRoundingButInefficient}. In Lemma \ref{lem:factorOfRoundingButInefficient}, it is sufficient to show that the uniform distribution \(\mu\) over subsets \(S\subseteq U\) of size \(k\) is \(\frac{k}{k-d+1}\)-approximate \((d-1,d)\)-wise independent  (as in \eqref{eq:with-rep-near-ind-distribution}). Here, it is sufficient to show that the  uniform distribution \(\mu\) is \(\frac{k}{k-l+1}\)-approximate \((l',l)\)-wise independent. For \(T,T'\subseteq\ [n]\) of size \(l',l\), the calculation of \(\frac{\Prob{\mu}{\altmathcal{S}\supseteq T'}}{\Prob{\mu}{\altmathcal{S}\supseteq T}}\) and \(\frac{y^{T'}}{y^{T}}\) is straightforward 
\begin{equation} \frac{\Prob{\mu}{\altmathcal{S}\supseteq T'}}{\Prob{\mu}{\altmathcal{S}\supseteq T}} = \frac{\binom{qk-l'}{k-l'}/\binom{qk}{k}}{\binom{qk-l}{k-l}/\binom{qk}{k}} \leq \frac{(qk)^{l-l'}(k-l)!}{(k-l')!} \ \text{ and } \ \frac{y^{T'}}{y^{T}}= q^{l-l'} \label{eq:with-rep-direct-calcuation2} \end{equation}
Therefore, \(\mu\) is \(\alpha\)-approximate \((l',l)\)-wise independent for\begin{align*}
\alpha&=\left(\frac{\Prob{\mu}{\altmathcal{S}\supseteq T'}}{\Prob{\mu}{\altmathcal{S}\supseteq T}}\cdot \frac{y^{T}}{y^{T'}} \right)^{\frac{1}{l-l'}} \leq\left(\frac{(qk)^{l-l'}(k-l)!}{(k-l')!} q^{l'-l}\right)^{\frac{1}{l-l'}} \\
&=\frac{k}{\left[(k-l')(k-l'-1)\cdots (k-l+1) \right]^{\frac{1}{l-l'}}}\leq\frac{k}{k-l+1}
\end{align*}
\cut{
 \begin{align*}
\alpha&=\frac{\binom{qk-l'}{k-l'}}{\binom{qk-l}{k-l}}\left(\frac{1}{q}\right)^{l-l'}\\
&=\frac{(qk-l')(qk-l'-1)\cdots(qk-l+1)}{\frac{(k-l')!}{(k-l)!}}\left(\frac{1}{q}\right)^{l-l'} \\
&\leq \frac{(qk)^{l-l'}}{\frac{(k-l')!}{(k-l)!}}\left(\frac{1}{q}\right)^{l-l'} \\
&=\frac{k^{l-l'}}{(k-l')(k-l'-1)\cdots (k-l+1)} \leq \left( \frac{k}{k-l+1}\right)^{l-l'}
\end{align*}}
as we wanted to show.
 \end{proof}
We note that the \(l\)-proportional volume sampling in the proof of Theorem \ref{thm:repetitions-Gen} can be implemented efficiently, and the proof is outlined in Section \ref{sec:eff-implementation-gen}.

\subsubsection{Integrality Gap}
Finally, we state an integrality gap for minimizing generalized ratio objective \(\left(\frac{E_{l'}(V_S V_S^\top)}{E_l(V_S V_S^\top)}\right)^{\frac{1}{l-l'}}\) over subsets \(S\subseteq [n]\) of size \(k\). The integrality gap matches our approximation ratio of our algorithm with repetitions when $k$ is large.
\begin{theorem} \label{thm:IntegralityGapLowerBound-Gen}
For any given positive integers \(k,d \) and a pair of integers $0\leq l'\leq l \leq d$ with \(k>l'\), there exists an instance \(V=[v_1,\ldots,v_n]\in\R^{d \times n}\) to the problem of minimizing \(\left(\frac{E_{l'}(V_S V_S^\top)}{E_l(V_S V_S^\top)}\right)^{\frac{1}{l-l'}}\) over subsets \(S\subseteq [n]\) of size \(k\)  such that 
\[ \OPT \geq \left(\frac{k}{k-l'}-\delta\right)\cdot \CP \]
for all $\delta >0,$ where $\OPT$ denotes the value of the optimal integral solution and $\CP$ denotes the value of the convex program.
\end{theorem}
This implies that the integrality gap is at least \( \frac{k}{k-l'} \) for minimizing\(\left(\frac{E_{l'}(V_S V_S^\top)}{E_l(V_S V_S^\top)}\right)^{\frac{1}{l-l'}}\)over subsets \(S\subseteq [n]\) of size \(k\). The theorem applies to both with and without repetitions.

\begin{proof}
The instance $V=[v_1,\ldots,v_n]$ will be the same for with and without repetitions. For each \(1\leq i \leq d\), let \(e_i\) denote the unit vector in the direction of axis \(i\). Choose
\begin{displaymath}
v_i=\begin{cases} \sqrt{N}\cdot e_i & \text{for }i=1,\ldots,l' \\
e_i  &  \text{for }i=1,\ldots,l'  \\
\end{cases}
\end{displaymath}
where \(N>0\) is a constant to be chosen later. Set $v_i,i>l$ to be at least $k$ copies of each of these $v_i$ for $i\leq l$, as we can make $n$ as big as needed. Hence, we may assume that we are allowed to pick only $v_i,i\leq l$, but with repetitions.

Let \(S^*\) represent the set of vectors in \(\OPT\) and  \(y_i\) be the number of copies of \(v_i\) in \(S^*\) for \(1\leq i \leq l\). Clearly \(y_i\geq 1\) for all \(i=1,\ldots,l\) (else the objective is unbounded). The eigenvalues of \(V_{S^*} V_{S^*}^\top\) are
\[\lambda(V_{S^*} V_{S^*}^\top) = ( y_1N,y_2N,\ldots,y_{l'}N,y_{l'+1},y_{l'+2},\ldots,y_l,0,\ldots,0)\] Hence, both  \(E_{l'}(V_{S^*} V_{S^*}^\top)=e_{l'}(\lambda)\) and \(E_l(V_{S^*} V_{S^*}^\top)=e_{l}(\lambda)\) are polynomials in variables \(N\) of degree \(l'\).

Now let \(N\rightarrow \infty\).  To compute \((\OPT)^{l-l'}=\frac{E_{l'}(V_{S^*} V_{S^*}^\top)}{E_{l}(V_{S^*} V_{S^*}^\top)}\), we only need to compute the coefficient of the highest degree monomial \(N^{l'}\). The
coefficient of \(N^{l'}\) in \(e_{l'}(\lambda),e_l(\lambda)\) are exactly \(\prod_{i=1}^{l'} y_i,\prod_{i=1}^{l} y_i\), and therefore
\[(\OPT)^{l-l'}=\frac{E_{l'}(V_{S^*} V_{S^*}^\top)}{E_l(V_{S^*} V_{S^*}^\top)}\rightarrow\frac{\prod_{i=1}^{l'} y_i}{\prod_{i=1}^{l} y_i}=\left(\prod_{i=l'+1}^{l} y_i\right)^{-1}\]
Observe that \(\prod_{i=l'+1}^{l} y_i\) is maximized under the budget constraint \(\sum_{i=1}^l y_i=|S^*|=k\) when \(y_j=1\) for \(j=1,\ldots,l'\). Therefore,
\begin{align*}
\prod_{i=l'+1}^{l} y_i \leq \left( \frac{1}{l-l'}\sum_{i=l'+1}^l y_i\right)^{l-l'}= \left( \frac{k-l'}{l-l'}\right)^{l-l'}
\end{align*}
where the inequality is by AM-GM. Hence, \(\OPT\) is lower bounded by a quantity that converges to \( \frac{l-l'}{k-l'}  \) as \(N\rightarrow \infty\).

We now give a valid fractional solution \(x\) to upper bound \(\CP\) for each \(N>0\). Choose
\begin{displaymath}
x_i = \begin{cases}  \frac{k}{\sqrt{N}} & \text{for } i=1,\ldots,l' \\
\frac{k-\frac{kl'}{\sqrt{N}}}{l-l'} & \text{for } i=l'+1,\ldots,l \\
0 & \text{for } i>l
\end{cases}
\end{displaymath}
Then, eigenvalues of \(V(x)V(x)^\top\) are
\begin{align*}
\lambda':=\lambda(V(x)V(x)^\top) &= ( x_1N,x_2N,\ldots,x_{l'}N,x_{l'+1},x_{l'+2},\ldots,x_l,0,\ldots,0) \\
&=(k\sqrt{N},k\sqrt{N},\ldots,k\sqrt{N},x_{l'+1},x_{l'+2},\ldots,x_l,0,\ldots,0)
\end{align*}
Now as \(N\rightarrow \infty\), the dominating terms of \(E_{l'}(V(x)V(x)^\top)=e_{l'}(\lambda')\) is \(\prod_{i=1}^{l'}(k\sqrt{N})=k^{l'}(\sqrt{N})^{l'}\). Also, we have
\begin{align*}E_{l}(V(x)V(x)^\top)&=e_{l}(\lambda')=\prod_{i=1}^{l'}(k\sqrt{N})\prod_{i=l'+1}^l x_i\\
&=k^{l'}\left(\frac{k-\frac{kl'}{\sqrt{N}}}{l-l'}\right)^{l-l'}(\sqrt{N})^{l'}\rightarrow k^{l'}\left(\frac{k}{l-l'}\right)^{l-l'}(\sqrt{N})^{l'}
\end{align*}
Hence, $$\CP \leq \left(\frac{E_{l'}(V(x)V(x)^\top)}{E_{l}(V(x)V(x)^\top)} \right)^{l-l'}\rightarrow \frac{l-l'}{k}$$ Therefore, \(\frac{\OPT}{\CP}\) is lower bounded by a ratio which converges to \(\frac{l-l'}{k-l'}\cdot\frac{k}{l-l'} = \frac{k}{k-l'}\).
\end{proof}

\section{Efficient Algorithms} \label{sec:EfficientAlgorithms}
In this section, we outline efficient sampling algorithms, as well as
deterministic implementations of our rounding algorithms, both for
with and without repetition settings.
\subsection{Efficient Randomized Proportional Volume}\label{Sec:EfficientRandomizedProportionalVolume}

Given a vector \(\lambda\in\R_+^n\), we show that proportional volume sampling with $\mu(S) \propto
\lambda^S$ for $S\in \U$, where \(\U\in\{\U_k,\U_{\leq k}\}\) can be
done in time polynomial in the size \(n\) of the ground set.
We start by stating a lemma which is very useful both for the sampling
algorithms and the deterministic implementations.

\begin{lemma} \label{sumofProductDet}
Let $\lambda\in\R_+^n,v_1,\ldots,v_n\in \R^d$, and $V=[v_1,\ldots,v_n]$. Let $I,J\subseteq [n]$ be disjoint. Let $1\leq k\leq n,0\leq d_0\leq d$. Consider the following function
\[ F(t_1,t_2,t_3)=\det\left(I_n +  t_1 \diag(y) + t_1t_2 \diag(y)^{1/2}VV^\top\diag(y)^{1/2} \right) \]
where $t_1,t_2,t_3\in \R$ are indeterminate, \(I_n\) is the \(n\times n\) identity matrix, and $y\in\R^n$ with
\[ y_i = \begin{cases}
      \lambda_it_3, & \text{if } i\in I \\
      0, & \text{if } i\in J \\
      \lambda_i, & \text{otherwise}
   \end{cases}. \]
Then $F(t_1, t_2, t_3)$ is a polynomial and the quantity
\begin{equation}\label{exp:hardcoreSum}
 \sum_{|S|=k,I\subseteq S,J\cap S=\emptyset}\lambda^S
 \sum_{|T|=d_0,T\subseteq S}\det(V_T^\top V_T) 
\end{equation}
is the coefficient of the monomial
$t_1^{k}t_2^{d_0}t_3^{|I|}$. 
Moreover, this quantity can be computed in \(O\pr{n^3d_0k |I|\cdot\log(d_0k|I|)}\) number of arithmetic operations. 
\end{lemma}

\begin{proof}
  Let us first fix some $S \subseteq [n]$. Then we have
  \[
  \sum_{|T|=d_0,T\subseteq S}\det(V_T^\top V_T) 
  = E_{d_0}(V_S^\top V_S)
  = 
  [t_2^{d_0}] \det(I_S + t_2 V_S V_S^\top),
  \]
  where the notation $[t_2^{d_0}]p(t_2)$ denotes the coefficient of
  $t^{d_0}$ in the polynomial $p(t_2) = \det(I_S + t_2 V_S
  V_S^\top)$. The first equality is just Cauchy-Binet, and
  the second one is standard and follows from the Leibniz
  formula for the determinant. Therefore, \eqref{exp:hardcoreSum}
  equals
  \[
  [t_2^{d_0}] \sum_{|S|=k,I\subseteq S,J\cap S=\emptyset}\lambda^S
  \det(I_S + t_2 V_S V_S^\top).
  \]

  To complete the proof, we establish the following claim.

  \begin{claim}\label{cl:cond-sum}
    Let $L$ be an $n\times n$ matrix, and let $\lambda, I, J, k, y$ be as in the
    statement of the Lemma. Then,
    \begin{align*}
    \sum_{|S|=k,I\subseteq S,J\cap S=\emptyset}\lambda^S \det(L_{S,S})
    &= 
    [t_3^{|I|}] E_k\left(\diag(y)^{1/2}L\ \diag(y)^{1/2}\right)\\
    &= 
    [t_1^kt_3^{|I|}] \det\left(I_n + t_1 \diag(y)^{1/2}L\ \diag(y)^{1/2}\right).
    \end{align*}
  \end{claim}
  \begin{proof}
    By Cauchy-Binet,
    \begin{align*}
      E_k\left(\diag(y)^{1/2}L\ \diag(y)^{1/2}\right) &=
      \sum_{|S| = k}y^S \det(L_{S,S})\\
      &= \sum_{|S| = k, J \cap S = \emptyset}t_3^{|S \cap I|}\lambda^S \det(L_{S,S}).
    \end{align*}
    The first equality follows. The second is, again, a consequence of
    the Leibniz formula for the determinant. 
  \end{proof}
  Plugging in $L = I_n + t_2 VV^\top$ in Claim~\ref{cl:cond-sum} gives
  that \eqref{exp:hardcoreSum} equals
  \begin{multline*}
  [t_1^kt_2^{d_0}t_3^{|I|}] 
  \det\left(I_n + t_1 \diag(y)^{1/2}(I_n + t_2 VV^\top) \diag(y)^{1/2}\right)\\
  =   [t_1^kt_2^{d_0}t_3^{|I|}] 
  \det\left(I_n +  t_1 \diag(y) + t_1t_2 \diag(y)^{1/2}VV^\top\diag(y)^{1/2}\right).
  \end{multline*}
  This completes the proof. For the running time, the standard computation time of matrix multiplication and determinant of \(n\times n\) matrices is \(O(n^3)\) entry-wise arithmetic operations. We need to keep all monomials in the form \(t_1^{a}t_2^{b}t_3^{c}\) where \(a\leq k, b \leq d_0, c \leq |I|\), of which there are \(O(d_0k|I|)\).
By representing multivariate monomials in single variable (\cite{pan1994simple}), we may use Fast Fourier Transform to do one polynomial multiplication of entries of the matrix in \(O\pr{d_0k|I|\cdot\log(d_0k|I|)}\) number of arithmetic operations. This gives the total running time of \(O\pr{n^3d_0k|I|\cdot\log(d_0k|I|)}\).
\end{proof}

Using the above lemma, we now prove the following theorem that will directly imply Lemma~\ref{lem:sample}.

\begin{theorem}
\label{thm:effic-randomized-hardcore}
Let $\lambda \in\R_+^n,v_1,\ldots,v_n\in \R^d,1\leq k\leq n$, \(\U\in\{ \U_k,\U_{\leq k} \} \), and $V=[v_1,\ldots,v_n]$. Then there is a randomized algorithm $\mathcal{A}$  which outputs $\altmathcal{S}\in\U$ such that
\[ \Prob{\mathcal{\SS\sim A}}{\altmathcal{S}=S}=\frac{\lambda^S \det(V_SV_S^\top)}{\sum_{S'\in \U}\lambda^{S'}\det(V_{S'}V_{S'}^\top)}=:\mu'(S) \]
That is, the algorithm correctly implements proportional volume sampling  \(\mu'\)
with hard-core measure \(\mu\) on \(\U\) with parameter \(\lambda \).
Moreover, the algorithm runs in \(O\pr{n^4dk^2\log(dk)}\) number of arithmetic operations. \end{theorem}

\begin{observation} \label{obs:n=k+d^2}
\cite{wang2016computationally} shows that we may assume that the support of an extreme fractional solution of convex relaxation has size at most \(k+d^2\). Thus, the runtime of proportional volume sampling is \(O\pr{(k+d^2)^4dk^2\log(dk)}\). While the degrees in \(d,k\) are not small, this runtime is independent of \(n\).
\end{observation}

\begin{observation}
It is true in theory and observed in practice that solving the continuous relaxation rather than the rounding algorithm is a bottleneck in computation time, as discussed in \cite{allen2017near}. In particular, solving the continuous relaxation of \textit{A}-optimal design takes \(O\pr{n^{2+\omega}\log n}\) number of iterations by standard ellipsoid method and \(O\pr{(n+d^2)^{3.5}}\) number of iterations by SDP, where \(O(n^\omega)\) denotes the runtime of \(n \times\ n\) matrix multiplication. In most applications where \(n >> k\), these running times dominates one of proportional volume sampling.

\end{observation}

\begin{proof}
We can sample by starting with an empty set \(S=\emptyset\). Then, in each step $i=1,2,\ldots,n$, the algorithm decides with the correct probability
\[ \Prob{\SS\sim \mu'}{i\in\altmathcal{S}|I\subseteq \SS, J\cap \SS=\emptyset} \]
whether to include $i$ in $S$ or not, given that we already know that we have included $I$ in $S$ and excluded $J$ from \(S\) from previous steps $1,2,\ldots,i-1$.
Let \(I'= I \cup \{i\}\). This probability equals to
\begin{align*}
\Prob{\SS\sim \mu'}{i\in\altmathcal{S}|I\subseteq \SS, J\cap \SS=\emptyset}  &= \frac{\Prob{\SS\sim\mu'}{I'\subseteq \SS, J\cap \SS=\emptyset} }{\Prob{\SS\sim\mu'}{I\subseteq \SS, J\cap \SS=\emptyset} }  \\
&= \frac{\sum_{S \in \U,I'\subseteq S,J\cap S=\emptyset}\lambda^{S} \det(V_SV_S^\top)}{\sum_{S \in \U,I\subseteq S,J\cap S=\emptyset}\lambda^{S} \det(V_SV_S^\top) } \\
&= \frac{\sum_{S \in \U,I'\subseteq S,J\cap S=\emptyset}\lambda^{S}\sum_{|R|=d,R\subset S}\det(V_RV_R^\top)}{\sum_{S \in \U,I\subseteq S,J\cap S=\emptyset}\lambda^{S}\sum_{|R|=d,R\subset S}\det(V_RV_R^\top)}
\end{align*}
where we apply the Cauchy-Binet formula in the last equality. For
\(\U=\U_k\), both the numerator and denominator are summations over
\(S\) restricted to \(|S|=k\), which can be computed in \(O\pr{n^3dk^2\log(dk)}\) number of arithmetic operations by Lemma~\ref{sumofProductDet}. For the case \(\U=\U_{\leq k}\), we can evaluate summations in the numerator and denominator restricted to \(|S|=k_0\)
 for each \(k_0=1,2,\ldots k\) by computing polynomial \(F(t_1,t_2,t_3)\) in Lemma~\ref{sumofProductDet} only once, and then sum those quantities over \(k_0\). 
 \end{proof}

\subsection{Efficient Deterministic Proportional  Volume} \label{sec:deterministic}

We show that for hard-core measures there is a
deterministic algorithm that achieves the same objective value as the
expected objective value achieved by proportional volume sampling. The
basic idea is to use the method of conditional expectations.

\begin{theorem}\label{thm:deterministic}
Let $\lambda\in\R_+^n,v_1,\ldots,v_n\in \R^d,1\leq k\leq n$,  \(\U\in\{ \U_k,\U_{\leq k} \} \), and $V=[v_1,\ldots,v_n]$. Then there is a deterministic algorithm \(\mathcal{A'}\) which outputs $S^*\subseteq [n]$ of size $k$ such that
\[ \tr\left(V_{S^*}V_{S^*}^\top\right)^{-1}\geq \Expectation{\mu'}{\tr\left(V_{\altmathcal{S}}V_{\altmathcal{S}}^\top\right)^{-1}} \]
where \(\mu'\) is the probability distribution defined by $\mu'(S) \propto
\lambda^S \det(V_SV_S^\top)$ for all $S\in \U$.
Moreover, the algorithm runs in \(O\pr{n^4dk^2\log(dk)}\) number of arithmetic operations. \end{theorem}
Again, with the assumption that \(n\leq k+d^2\) (Observation \ref{obs:n=k+d^2}), the runtime for deterministic proportional volume sampling is \(O\pr{(k+d^2)^4dk^2\log(dk)}\).

\begin{proof}
To prove the theorem, we derandomize the sampling algorithm in Theorem \ref{thm:effic-randomized-hardcore} by the method of conditional expectations. The deterministic algorithm starts with \(S^*=\emptyset \), and then chooses, at each step $i=1,2,\ldots,n$, whether to pick $i$ to be in $S^*$ or not, given that we know from previous steps to include or exclude each element $1,2,\ldots,i-1$ from $S^*$. The main challenge is to calculate exactly the quantity of the form
\[ X(I,J):=\Expectation{\SS\sim\mu'}{\tr\left(V_\altmathcal{S}V_\altmathcal{S}^\top\right)^{-1}|I\subset \altmathcal{S},J\cap \altmathcal{S}=\emptyset} \]
where \(I,J\subseteq[n] \) are disjoint. If we can efficiently calculate the quantity of such form, the algorithm can, at each step \(i=1,2,\ldots,n\), calculate \(X(I'\cup\{i\},J') \) and \(X(I',J'\cup\{i\}) \)
where \(I',J'\subseteq[i-1]\) denote elements we have decided to pick and not to pick, respectively, and then  include \(i\) to \(S^*\) if and only if \(X(I'\cup\{i\},J')\geq X(I',J'\cup\{i\}). \)

 Note that the quantity \(X(I,J)\) equals
\begin{align*}
\Expectation{\SS\sim\mu'}{\tr\left(V_\altmathcal{S}V_\altmathcal{S}^\top\right)^{-1}|I\subset \altmathcal{S},J\cap \altmathcal{S}=\emptyset} &= \sum_{\substack{S \in \U,\\I\subseteq S, J\cap S=\emptyset}} \Prob{\mu'}{\altmathcal{S}=S|I \subseteq \altmathcal{S},\altmathcal{S}\cap J=\emptyset}\tr\left[(V_SV_S^\top)^{-1}\right] \\
&= \sum_{\substack{S \in \U,\\I\subseteq S, J\cap S=\emptyset}} \frac{\lambda^S\det(V_SV_S^\top)}{\sum_{S' \in \U,I\subseteq S,J\cap S=\emptyset}\lambda^{S'}\det(V_{S'}V_{S'}^\top)}\tr\left[(V_SV_S^\top)^{-1}\right] \\
&= \frac{\sum_{S \in \U,I\subseteq S,J\cap S=\emptyset}\lambda^{S} E_{d-1}( V_SV_S^\top)}{\sum_{S \in \U,I\subseteq S,J\cap S=\emptyset}\lambda^{S}\sum_{|R|=d,R\subset S}\det(V_RV_R^\top)} \\
&= \frac{\sum_{S \in \U,I\subseteq S,J\cap S=\emptyset}\lambda^{S}\sum_{|T|=d-1,T\subset S}\det(V_T^\top V_T)}{\sum_{S \in \U,I\subseteq S,J\cap S=\emptyset}\lambda^{S}\sum_{|R|=d,R\subset S}\det(V_RV_R^\top)}
\end{align*}
where we  write inverse of trace as ratio of symmetric polynomials of eigenvalues in the third equality and use
Cauchy-Binet formula for the third and the
fourth equality. 
The rest of the proof is now identical to the proof of Theorem \ref{thm:effic-randomized-hardcore}, except with different parameters \(d_0=d-1,d\) in \(f(t_1,t_2,t_3)\) when applying Lemma \ref{sumofProductDet}. 
\cut{For
\(\U=\U_k\), both the numerator and denominator are summations over
\(S\) restricted to \(|S|=k\), which can be computed in \(O\pr{n^3dk^2\log(dk)}\) number of arithmetic operations by Lemma~\ref{sumofProductDet}. For the case \(\U=\U_{\leq k}\), we can evaluate summations in the numerator and denominator restricted to \(|S|=k_0\)
 for each \(k_0=1,2,\ldots k\) by computing polynomial \(F(t_1,t_2,t_3)\) in Lemma~\ref{sumofProductDet} only once, and then sum those quantities over \(k_0\).}
\end{proof}

\subsection{Efficient Randomized Implementation of $\frac{k}{k-d+1}$-Approximation Algorithm With Repetitions }\label{sec:effreplace}

First, we need to state several Lemmas needed to compute particular sums. The main motivation that we need a different method from Section \ref{Sec:EfficientRandomizedProportionalVolume} and    \ref{sec:deterministic} to compute a similar sum is that we want to allow the ground set \(U\)
of indices of all copies of vectors to have an exponential size. This makes Lemma \ref{sumofProductDet} not useful, as the matrix needed to be computed has dimension \(|U|\times|U|\).
The main difference, however, is that the parameter \(\lambda\) is now a constant, allowing us to obtain sums by computing a more compact \(d\times d \) matrix.
\begin{lemma} \label{lem:computeSize-d-SumDeternimantFixedIntersection}
Let $V=[v_1,\ldots,v_m]$ be a matrix of vectors $v_i\in\R^d$ with \(n\geq d\) distinct vectors. Let $F\subseteq [m]$ and let $0\leq r \leq d$ and \(0\leq d_0 \leq d\) be integers. Then the quantity $\sum_{|T|=d_0, |F \cap R|=r} \det(V_T^\top V_T)$ is the coefficient of $t_1^{d-d_0} t_2^{d_0-r}t_3^{r}$ in
\begin{equation}
f(t_1,t_{2},t_3)=\det\left( t_1I_d +\sum_{i\in F}t_3v_iv_i^\top + \sum_{i\notin F}t_2 v_iv_i^\top \right) \label{determinantSumCoefficientComparison}
\end{equation}
where $t_1,t_2,t_3\in\R$ are indeterminate and \(I_d\) is the  \(d\times d\)
identity matrix. Furthermore, this quantity can be computed in \(O\pr{n(d-d_0+1)d_0^2d^2\log d}\) number of arithmetic operations. 
\end{lemma}

\begin{proof}
First,  note that  \(\det\left( t_1I +\sum_{i\in F}t_3v_iv_i^\top +
  \sum_{i\notin F}t_2 v_iv_i^\top \right)=\prod_{i=1}^d (t_1+\nu_i) \)
where \(\nu(M)=\{\nu_1,\ldots, \nu_d\}\) is the vector of eigenvalues
of the matrix $M=\sum_{i\in F}t_3v_iv_i^\top + \sum_{i\notin F}t_2
v_iv_i^\top $. Hence, the coefficient of \(t_1^{d-d_0}\) in
$\det\left( t_1I +\sum_{i\in F}t_3v_iv_i^\top + \sum_{i\notin F}t_2 v_iv_i^\top \right)$ is \(e_{d_0}(\nu(M)). \)

Next, observe that \(M\) is in the form $V'V'^\top$ where $V'$ is the
matrix where columns are $\sqrt{t_3}v_i$, $i\in F$ and
$\sqrt{t_2}v_i,i\notin F$. Applying Cauchy-Binet to
$E_{d_0}(V'V'^\top)$, we get
\begin{align*}
E_{d_0}\left( \sum_{i\in F}t_3v_iv_i^\top + \sum_{i\notin F}t_2 v_iv_i^\top \right)=E_{d_0}(V'V'^\top) &= \sum_{|T|=d_0}\det(V_T'^\top V_T') \\
&= \sum_{l=0}^{|F|} \sum_{|T|=d_0, |T\cap F|=l} \det(V_T'^\top V_T') \\
&= \sum_{l=0}^{|F|} \sum_{|T|=d_0, |T\cap F|=l} t_3^l t_2^{d_0-l} \det(V_T^\top V_T),
\end{align*}
where we use Lemma~\ref{lemma:LinearlityofDeterminant} for the last
equality. The desired quantity \(\sum_{|T|=d_0, |F \cap R|=r} \det(V_T^\top V_T)\) is then  exactly the coefficient at
$l=r$ in the sum on the right hand side.

To compute the running time, since there are only \(n\) distinct vectors, we may represent sets \(V,F\) compactly with distinct \(v_i\)'s and number of copies of each distinct \(v_i\)'s. Therefore, computing the matrix sum takes \(O\pr{nd^2}\) entry-wise  operations.
Next, the standard computation time of determinant of \(d\times d\) matrix is \(O(d^3)\) entry-wise arithmetic operations. This gives a total  of \(O\pr{nd^2+d^3}=O\pr{nd^2}\) entry-wise operations.

For each entry-wise  operation, we keep all monomials in the form \(t_1^{a}t_2^{b}t_3^{c}\) where \(a\leq d-d_0, b \leq d_0-r, c \leq r\), of which there are \(O((d-d_0+1)d_0^2)\).
By representing multivariate monomials in single variable (\cite{pan1994simple}) of degree \(O((d-d_0+1)d_0^2)\), we may use Fast Fourier Transform to do one polynomial multiplication of entries of the matrix in \(O\pr{(d-d_0+1)d_0^2\log d}\) number of arithmetic operations. This gives the total runtime  of \(O\pr{n(d-d_0+1)d_0^2d^2\log d}\) arithmetic operations.
\end{proof}


\begin{lemma} \label{lem:computeSumDetContainingFixedSet}
Let $V=[v_1,\ldots,v_m]$ be a matrix of vectors $v_i\in\R^d$ with \(n\geq d\) distinct vectors. Let $F\subseteq [m]$ and let $0\leq r \leq d$ and \(0\leq d_0 \leq d\) be integers. There is an   algorithm to compute $\sum_{|S|=k,S\supseteq F} E_{d_0}(V_SV_S^\top) $ with \(O\pr{n(d-d_0+1)d_0^2d^2\log d}\) number of arithmetic operations.

\cut{To be precise,
\[\sum_{|S|=k,S\supseteq F} \det(V_SV_S^T) = \sum_{r=0}^d \binom{n-|F|-d+r}{k-|F|-d+r}  \text{Coef}\left(\det\left( \sum_{i\in F}t_3v_iv_i^\top + \sum_{i\notin F}t_2 v_iv_i^\top \right),t_2^{d-r}t_3^r \right)\]
hen the computation time is polynomial in $n,k,d$.}
\end{lemma}
\begin{proof}
 We apply Cauchy-Binet:\begin{align*}
 \sum_{|S|=k,S\supseteq F} E_{d_0}(V_SV_S^T) &=  \sum_{|S|=k,S\supseteq F} \sum_{|T|=d_0, T\subset S} \det(V_T^\top  V_T)\\
\cut{ &= \sum_{|T|=d}\det(V_TV_T^\top) \sum_{|S|=k,S\supset T,S\supset F} 1 \\}
 &= \sum_{|T|=d_0}\det(V_T^\top  V_T) \binom{m-|F|-d_0+|F\cap T|}{k-|F|-d_0+|F\cap T|} \\
\cut{ &= \sum_{r=0}^d \sum_{|T|=d, |F \cap T|=r} \det(V_TV_T^\top) \binom{n-|F|-d+r}{k-|F|-d+r} \\}
 &=  \sum_{r=0}^d \binom{m-|F|-d_0+r}{k-|F|-d_0+r}  \sum_{|T|=d_0, |F \cap T|=r} \det(V_T^\top  V_T )
\end{align*}
where we change the order of summations for the second equality,
and  enumerate over possible sizes of $F\cap T$ to get the third
equality.  We compute $f(t_1,t_{2},t_3)$ in Lemma \ref{lem:computeSize-d-SumDeternimantFixedIntersection} once with \(O\pr{n(d-d_0+1)d_0^2d^2\log d}\) number of arithmetic operations, so we obtain values of $\sum_{|T|=d_0, |F \cap T|=r} \det(V_T^\top V_T) $  for all $r=0,\ldots,d_0$. The rest is a straightforward calculation. \cut{The quantity \(\binom{n-|F|-d_0+r}{k-|F|-d_0+r}\) which is at most \(n^k=2^{k\log n}\) can be computed in \(\poly(k,\log n)\) time for each \(r=0,\ldots,d\). So the sum
\[ \sum_{r=0}^d \binom{n-|F|-d_0+r}{k-|F|-d_0+r}  \sum_{|T|=d_0, |F \cap T|=r} \det(V_T^\top V_T ) \]
over \(d\) terms can be computed in \(\poly(\log n,d,k)\) time}
\end{proof}

We  now present an efficient sampling procedure for Algorithm \ref{alg:duplicateSampleSamex_i}. We want to sample $S$ proportional to $\det(W_SW_S^\top)$. The set $S$ is a subset of all copies of at most $n$ distinct vectors, and there can be exponentially many copies. However, the key is that the quantity \(f(t_1,t_2,t_3) \) in \eqref{determinantSumCoefficientComparison} is still efficiently computable because exponentially many of these copies of vectors are the same.

\begin{theorem}
\label{Thm:effic-randomized-hardcore}
Given  inputs  $n,d,k,\epsilon ,x\in \R_+^n$ with $\sum_{i=1}^n x_i = k$, and vectors $v_1,\ldots,v_n$ \(\)\(\) to Algorithm \ref{alg:duplicateSampleSamex_i} we define \(q,U,W\) as in Algorithm \ref{alg:duplicateSampleSamex_i}. Then, there exists an implementation \(\A\) that samples \(\SS\) from the distribution \(\mu'\) over all subsets \(S\subseteq\ U\) of size \(k\), where \(\mu'\) is defined by $\Pr_{\SS\sim\mu'}[\SS=S]\propto \det(W_{S}W_S^\top)$ for each  $S\subseteq U,|S|=k$. Moreover, \(\A\) runs in  \(O\pr{n^2d^4k\log d}\) number of arithmetic operations.\end{theorem}
Theorem \ref{Thm:effic-randomized-hardcore} says that steps  \eqref{step:copyingStepAlgRepetition}-\eqref{step:samplingStepAlgRepetition} in Algorithm \ref{alg:duplicateSampleSamex_i} can be efficiently implemented. Other steps except \eqref{step:copyingStepAlgRepetition}-\eqref{step:samplingStepAlgRepetition} obviously use  \(O\pr{n^2d^4k\log d}\) number of arithmetic operations, so the above statement  implies that Algorithm \ref{alg:duplicateSampleSamex_i} runs in  \(O\pr{n^2d^4k\log d}\) number of arithmetic operations.
Again, by Observation \ref{obs:n=k+d^2}, the number of arithmetic operations is in fact \(O\pr{(k+d^2)^2d^4k\log d}\).
 
\begin{proof}
Let $m_i = qx_i'$ be the number of copies of vector $v_i$ (recall that \(q=\frac{2n}{\epsilon k}\))\cut{ and \(m=\sum_{i=1}^n m_i\)}. \cut{We may assume that $q>k$ by picking smaller \(\epsilon\).} Let $w_{i,j}$ denote the $j$th copy of vector $v_i$. Write $U=\{(i,j):i\in[n],j\in[m_i]\}$ be the new set of indices after the copying procedure. Denote $\SS$  a random subset (not multiset) of $U$ that we want to sample. Write $W$ as the matrix with columns $w_{i,j}$ for all \((i,j)\in U\). Let $E_i= \{w_{i,j}:j=1,\ldots,m_i\}$ be the set of copies of vector \(v_{i}\). For any \(A\subseteq U\), we say that \textit{$A$ has $k_i$ copies of $v_i$} to mean that $|A\cap E_i| = k_i$.

We can define the sampling algorithm \(\altmathcal{A}\) by sampling, at each step $t=1,\ldots,n$, how many copies of \(v_i\) are to be included in \(\altmathcal{S}\subseteq U\).
Denote \(\mu'\) the volume sampling on \(W\) we want to sample. The problem then reduces to efficiently computing
\begin{align}
\Prob{\mu'}{\altmathcal{S} \text{ has } k_t \text{ copies of } v_t |  \altmathcal{S} \text{ has } k_i \text{ copies of } v_i,\forall  i=1,\ldots,t-1} \nonumber \\
= \frac{\Prob{\mu'}{\altmathcal{S} \text{ has } k_i \text{ copies of } v_i,\forall  i=1,\ldots,t}}{\Prob{\mu'}{\altmathcal{S} \text{ has } k_i \text{ copies of } v_i, \forall  i=1,\ldots,t-1}} \label{eq:without-rep-deter-quantity}
\end{align}
for each \(k_t=0,1,\ldots,k-\sum_{i=1}^{t-1}k_i\). Thus, it suffices to efficiently compute quantity \eqref{eq:without-rep-deter-quantity} for any given $1\leq t\leq n$ and $k_1,\ldots,k_t$ such that $\sum_{i=1}^t k_i \leq k$.

We now fix \(t,k_1,\ldots,k_t\). Note that for any $i\in[n]$, getting any set of $k_i$ copies of $v_i$ is the same, i.e.  events $\altmathcal{S}\cap E_i= F_i$ and $\altmathcal{S}\cap E_i= F_i'$ under \(\SS\sim\mu'\) have the same probability for any subsets $F_i,F_i'\subseteq E_i$ of the same size.
Therefore, we fix one set of $k_i$ copies of $v_i$ to be $F_i=\{w_{i,j}:j=1,\ldots,k_i\}     $ for all \(i\in[n]\) and obtain
\[ \Prob{}{\altmathcal{S} \text{ has } k_i \text{ copies of } v_i, \forall  i=1,\ldots,t} =\prod_{i=1}^t\binom{m_i}{k_i} \Prob{}{\altmathcal{S}\cap E_i = F_i,\forall i=1,\ldots t}\]
Therefore, \eqref{eq:without-rep-deter-quantity} equals
\begin{align}
&\frac{\prod_{i=1}^t\binom{m_i}{k_i} \Prob{}{\altmathcal{S}\cap E_i = F_i,\forall i=1,\ldots t}}{\prod_{i=1}^{t-1}\binom{m_i}{k_i} \Prob{}{\altmathcal{S}\cap E_i = F_i,\forall i=1,\ldots t-1}}= \binom{m_t}{k_t}\frac{\sum_{|S|=k,S\cap E_i = F_i,\forall i=1,\ldots t}\det(W_SW_S^\top)}{\sum_{|S|=k,S\cap E_i = F_i,\forall i=1,\ldots t-1}\det(W_SW_S^\top)} \label{eq:WithRepFractionEachStep}
\end{align}
To compute the numerator, define \(W'\) a matrix of vectors in \(W\) restricted to indices \(U\setminus\left(\bigcup_{i=1}^t E_i\setminus F_i\right)\), and $F:=\bigcup_{i=1}^t F_i$, then we have
\begin{equation}
\sum_{|S|=k,S\subseteq W,S\cap E_i = F_i,\forall i=1,\ldots t}\det(W_SW_S^\top) = \sum_{|S|=k,S\subseteq W',S\supseteq F}\det(W'_S{W_S'}^\top) \label{eq:WithRepSumCauchy}
\end{equation}
By Lemma  \ref{lem:computeSumDetContainingFixedSet}, the
number of arithmetic operations to compute \eqref{eq:WithRepSumCauchy} is \(O\pr{n(d-d_0+1)d_0^2d^2\log d}=O\pr{nd^4\log d}\) (by applying \(d_0=d\)). Therefore, because in each step \(t=1,2,\ldots,n\), we compute \eqref{eq:without-rep-deter-quantity} at most \(k\) times for different values of \(k_t\), the total number of arithmetic operations for  sampling algorithm \(\A\) is \(O\pr{n^2d^4k\log d}\).  
\end{proof}
\begin{remark} \label{rem:run-time-size-with-eps}
Although Theorem \ref{Thm:effic-randomized-hardcore} and Observation \ref{obs:n=k+d^2} imply that randomized rounding for \textit{A}-optimal design with repetition  takes \(O\pr{(k+d^2)^2d^4k\log d}\) number of arithmetic operations, this does not take into account the size of numbers used in the computation which may scale with input \(\epsilon\). It is not hard to see that the sizes of coefficients \(f(t_1,t_2,t_3)\) in Lemma  \ref{lem:computeSize-d-SumDeternimantFixedIntersection}, of the number \(\binom{m-|F|-d_0+r}{k-|F|-d_0+r}\) in the proof of Lemma \ref{lem:computeSumDetContainingFixedSet}, and of  \(\binom{m_t}{k_t}\) in \eqref{eq:WithRepFractionEachStep} scale linearly with \(O(k\log\pr{m})\) where \(m=\sum_{i=1}^nm_i\). As we apply \(m\leq qk=\frac{2n}{\epsilon}\) in the proof of Theorem \ref{Thm:effic-randomized-hardcore}, the runtime of randomized rounding for \textit{A}-optimal design with repetition, after taking into account the size of numbers in the computation, has an extra factor of \(k\log(\frac{n}{\epsilon})\) and becomes \(O\pr{(k+d^2)^2d^4k^2\log d\log(\frac{k+d^2}{\epsilon}))}\).
\end{remark}
\subsection{Efficient Deterministic Implementation of $\frac{k}{k-d+1}$-Approximation Algorithm With Repetitions} \label{Sec:deterministicAlgwithRep}

We show  a \textit{deterministic} implementation of proportional volume sampling used for the  $\frac{k}{k-d+1}$-approximation algorithm with repetitions. In particular, we derandomized the efficient implementation of steps \eqref{step:copyingStepAlgRepetition}-\eqref{step:samplingStepAlgRepetition} of Algorithm \ref{alg:duplicateSampleSamex_i}, and show that the running time of deterministic version is the same as that of the randomized one.

\begin{theorem}\label{thm:deterministicWithRepetition}
Given  inputs  $n,d,k,\epsilon ,x\in \R_+^n$ with $\sum_{i=1}^n x_i = k$, and vectors $v_1,\ldots,v_n$ \(\)\(\) to Algorithm \ref{alg:duplicateSampleSamex_i}, we define \(q,U,W\) as in Algorithm \ref{alg:duplicateSampleSamex_i}. Then, there exists a deterministic algorithm \(\mathcal{A'}\) that outputs $S^*\subseteq U$ of size $k$ such that
\[ \tr\left(W_{S^*}W_{S^*}^\top\right)^{-1}\geq \Expectation{\SS \sim \mu' }{\tr\left(W_{\altmathcal{S}}W_{\altmathcal{S}}^\top\right)^{-1}} \]
where \(\mu'\) is a distribution over all subsets \(S\subseteq U\) of size \(k\) defined by \(\mu'(S)\)  \(\propto\det(W_SW_S^\top) \) for each set \(S\subseteq U\) of size \(k\). Moreover, \(\A'\) runs in  \(O\pr{n^2d^4k\log d}\) number of arithmetic operations.
\end{theorem}
Again, together with Observation \ref{obs:n=k+d^2} and Remark \ref{rem:run-time-size-with-eps}, Theorem \ref{thm:deterministicWithRepetition} implies that the $\frac{k}{k-d+1}$-approximation algorithm for \textit{A}-optimal  design with repetitions can be implemented deterministically in \(O\pr{(k+d^2)^2d^4k\log d}\) number of arithmetic operations  and, after taking into account the size of numbers in the computation, in \(O\pr{(k+d^2)^2d^4k^2\log d\log(\frac{k+d^2}{\epsilon})}\) time.
     
\begin{proof}
We can define the deterministic algorithm \(\altmathcal{A'}\) by deciding, at each step $t=1,\ldots,n$, how many copies of \(v_i\) are to be included in \(S^*\subseteq U\).
The problem then reduces to efficiently computing
\begin{equation}
\label{eq:expectationDeterministicWithRepetition}
X(k_1,\ldots,k_t):=\Expectation{\mu'}{\tr\left(W_\altmathcal{S}W_\altmathcal{S}^\top\right)^{-1}|  \altmathcal{S} \text{ has } k_i \text{ copies of } v_i,\forall  i=1,\ldots,t-1,t}
\end{equation}
where \(k_1,\ldots , k_{t-1}\) is already decided by previously steps of the algorithm, and now we compute \eqref{eq:expectationDeterministicWithRepetition} for each  \(k_t =0,1,\ldots, k-\sum_{i=1}^{t-1} k_i\).
\(\mathcal{A}'\) then chooses value of \(k_t\) which maximizes \eqref{eq:expectationDeterministicWithRepetition} to complete  step \(t\).

Recall the definitions from proof of Theorem \ref{Thm:effic-randomized-hardcore}  that \(F_i,E_i\) are the sets of fixed $k_i$ copies and all copies of \(v_i\), respectively,   $W'$ is the matrix of vectors in \(W\) restricted to indices $U\setminus\left(\bigcup_{i=1}^t E_i\setminus F_i\right)$, and $F:=\bigcup_{i=1}^t F_i$. Consider that

\begin{align*}
X(k_1,\ldots,k_t) &= \sum_{\substack{S \subseteq U;|S|=k; \\ |S\cap E_i|=k_i,\forall i=1,\ldots,t }} \Prob{\mu'}{\altmathcal{S}=S|  \altmathcal{S} \text{ has } k_i \text{ copies of } v_i, \forall  i=1,\ldots,t}\tr\left[(W_SW_S^\top)^{-1}\right] \\
&= \sum_{\substack{S \subseteq U;|S|=k; \\ |S\cap E_i|=k_i,\forall i=1,\ldots,t }} \frac{\det(W_SW_S^\top)}{\sum_{S' \subseteq U;|S'|=k;|S'\cap E_i|=k_i,\forall  i=1,\ldots,t}\det(W_{S'}W_{S'}^\top)}\tr\left[(W_SW_S^\top)^{-1}\right] \\
&=  \frac{\sum_{S \subseteq U;|S|=k;|S\cap E_i|=k_i,\forall  i=1,\ldots,t}E_{d-1} ( W_SW_S^\top)}{\sum_{S \subseteq U;|S|=k;|S\cap E_i|=k_i,\forall  i=1,\ldots,t}\det(W_{S}W_{S}^\top)}
\\
&= \frac{\prod_{i=1}^t\binom{m_i}{k_i}\sum_{S \subseteq U;|S|=k;S\supseteq F} E_{d-1} ( W'_S{W_S'}^\top)}{\prod_{i=1}^t\binom{m_i}{k_i}\sum_{S \subseteq U;|S|=k;S\supseteq F} \det( W'_S{W_S'}^\top)} \\
&= \frac{\sum_{S \subseteq U;|S|=k;S\supseteq F} E_{d-1} (  W'_S{W_S'}^\top)}{\sum_{S \subseteq U;|S|=k;S\supseteq F} \det( W'_S{W_S'}^\top)}
\end{align*}

By Lemma  \ref{lem:computeSumDetContainingFixedSet},
we can compute the numerator and denominator in \(O\pr{n(d-d_0+1)d_0^2d^2\log d}=O\pr{nd^4\log d}\) (by applying \(d_0=d-1,d\)) number of arithmetic operations. Therefore, because in each step \(t=1,2,\ldots,n\), we compute \eqref{eq:expectationDeterministicWithRepetition} at most \(k\) times for different values of \(k_t\), the total number of arithmetic operations for  sampling algorithm \(\A\) is \(O\pr{n^2d^4k\log d}\).
\end{proof}

\subsection{Efficient Implementations for the Generalized Ratio Objective} \label{sec:eff-implementation-gen}
In Section \ref{Sec:EfficientRandomizedProportionalVolume}-\ref{sec:deterministic} we obtain efficient randomized and deterministic implementations of proportional volume sampling with measure \(\mu \) when \(\mu\) is a hard-core distribution over all subsets \(S\in \U\) (where \(\U\in\{ \U_k,\U_{\leq k} \} \)) with any given parameter \(\lambda\in\R_+^n\). Both implementations run in \(O\pr{n^4dk^2\log(dk)}\) number of arithmetic operations. In Section \ref{sec:effreplace}-\ref{Sec:deterministicAlgwithRep}, we obtain efficient randomized and deterministic implementations of proportional volume sampling over exponentially-sized matrix \(W=[w_{i,j}]\) of \(m\) vectors  containing \(n\) \textit{distinct} vectors in \(O\pr{n^2d^4k\log d}\) number of arithmetic operations.\cut{\footnote{In  with-repetition setting, we also require \(V\) to contain at least \(k\) copies  for each distinct vector, but it is easy to see that the efficient randomized and deterministic implementations can still be obtained if this requirement is dropped by adjusting the range of possible \(k_i\) used in the proofs.}} In this section, we show that the results from Section  \ref{Sec:EfficientRandomizedProportionalVolume}-\ref{Sec:deterministicAlgwithRep} generalize to proportional \(l\)-volume sampling for generalized ratio problem. \begin{theorem}
\label{thm:effic-randomized-hardcore-Gen}
Let \(n,d,k\) be positive integers, $\lambda \in\R_+^n$, \(\U\in\{ \U_k,\U_{\leq k} \} \), $V=[v_1,\ldots,v_n]\in\R^{d \times n}$, and \(0\leq l' < l \leq d\) be a pair of integers. Let \(\mu'\) be the \(l\)-proportional volume sampling distribution over \(\U\) with hard-core measure \(\mu\) of parameter $\lambda$, i.e. \(\mu'(S)\propto\lambda^SE_l\pr{V_SV_S^\top}\) for all \(S\in \U\). There are

\begin{itemize}
\item {}an implementation to sample from \(\mu'\) that  runs in \(O\pr{n^4lk^2\log(lk)}\) number of arithmetic operations, and
\item a deterministic algorithm that outputs a set $S^*\in \U$ of size \(k\) such that
\begin{equation}
\left(
\frac{E_{l'}(V_{S^*} V_{S^*}^\top)}{E_l(V_{S^*} V_{S^*}^\top)} \right)^{\frac{1}{l-l'}}\geq \Expectation{\SS\sim\mu'}{\left(\frac{E_{l'}(V_\SS V_\SS^\top)}{E_l(V_\SS V_\SS^\top)}\right)^{\frac{1}{l-l'}}} \label{eq:deterministic-gen-ineq}
\end{equation}
that  runs in \(O\pr{n^4lk^2\log(lk)}\) number of arithmetic operations.
\end{itemize}
Moreover, let \(W=[w_{i,j}]\) be a matrix of \(m\) vectors where \(w_{i,j}=v_i\) for all \(i\in[n]\) and \(j\). Denote \(U\) the index set of \(W\). Let \(\mu'\) be the \(l\)-proportional volume sampling over all subsets \(S\subseteq U\) of size \(k\) with  measure \(\mu\) that is uniform, i.e. \(\mu'(S)\propto E_l\pr{W_SW_S^\top}\) for all \(S \subseteq U,|S|=k\).
There are
\begin{itemize}
\item
an implementation to sample from \(\mu'\) that  runs in \(O\pr{n^2(d-l+1)l^2d^2k\log d}\) number of arithmetic operations,
and\item
 a deterministic algorithm  that outputs a set $S^*\in \U$ of size $k$ such that
\begin{equation}
\left(
\frac{E_{l'}(W_{S^*} W_{S^*}^\top)}{E_l(W_{S^*} W_{S^*}^\top)} \right)^{\frac{1}{l-l'}}\geq \Expectation{\SS\sim\mu'}{\left(\frac{E_{l'}(W_\SS W_\SS^\top)}{E_l(W_\SS W_\SS^\top)}\right)^{\frac{1}{l-l'}}} \label{eq:deterministic-gen-ineq-with-rep}
\end{equation}
that  runs in \(O\pr{n^2\pr{(d-l'+1)l'^2+(d-l+1)l^2}d^2k\log d}\) number of arithmetic operations.
\end{itemize}
\end{theorem}

As in Observation \ref{obs:n=k+d^2}, note that  we can replace \(n=k+d^2\) in all running times in Theorem \ref{thm:effic-randomized-hardcore-Gen}  so that running times of all variants of proportional volume sampling are independent of \(n\).
We also note, as in Remark \ref{rem:run-time-size-with-eps}, that running times of \(l\)-proportional volume sampling over \(m\) vectors with \(n\) distinct vectors has an extra factor of \(k\log m\) after taking into account the size of numbers in computation, allowing us to do sampling over exponential-sized ground set \([m]\). 

\begin{proof}
By the convexity of \(f(z)=z^{l-l'}\) over positive reals \(z\), we have \(\Expectation{}{X}\geq \left(\Expectation{}{X^{\frac{1}{l-l'}} } \right)^{l-l'} \) for a nonnegative random variable \(X\). Therefore, to show \eqref{eq:deterministic-gen-ineq}, it is sufficient to show that \begin{equation}
\frac{E_{l'}(V_{S^*} V_{S^*}^\top)}{E_l(V_{S^*} V_{S^*}^\top)} \geq \Expectation{\SS\sim\mu'}{\frac{E_{l'}(V_\SS V_\SS^\top)}{E_l(V_\SS V_\SS^\top)}}
\end{equation} That is, it is enough to derandomized with respect to the objective
\(\frac{E_{l'}(V_{S} V_{S}^\top)}{E_l(V_{S} V_{S}^\top)}\), and the same is true for showing \eqref{eq:deterministic-gen-ineq-with-rep}. Hence, we choose to calculate the conditional expectations with respect to this objective.  

We follow the exact same calculation for \(l\)-proportional volume sampling for generalized ratio objective as original  proofs of efficient implementations of all four algorithms in \(A\)-optimal objective. We observe that those proofs in \(A\)-optimal objective ultimately rely on the ability to, given disjoint \(I,J\subseteq [n]\) (or in the other case, \([m]\)),  efficiently compute
\begin{align*}
\sum_{S \in \U,I\subseteq S,J\cap S=\phi}\lambda^{S}\sum_{|R|=d,R\subseteq S}\det(V_RV_R^\top)\text{ and }\sum_{S \in \U,I\subseteq S,J\cap S=\phi}\lambda^{S}\sum_{|T|=d-1,T\subseteq S}\det(V_T^\top V_T)
\end{align*}  (or in the other case, replace \(V\) with \(W\) and \(\lambda^S=1\) for all \(S\)). The proofs for generalized ratio objective  follow the same line as those proofs of four algorithms, except that we instead need to efficiently compute
\begin{align*}
\sum_{S \in \U,I\subseteq S,J\cap S=\phi}\lambda^{S}\sum_{|T|=l,R\subseteq S}\det(V_T^\top V_T)\text{ and }\sum_{S \in \U,I\subseteq S,J\cap S=\phi}\lambda^{S}\sum_{|T'|=l',T'\subseteq S}\det(V_{T'}^\top V_{T'})
\end{align*} (note the change of \(R,T\) of size \(d,d-1\) to \(T,T'\) of size \(l,l'\) respectively). But the computations can indeed be done efficiently by using different \(d_0=l',l\) instead of \(d_0=d-1,d\) when applying Lemmas \ref{sumofProductDet}, \ref{lem:computeSize-d-SumDeternimantFixedIntersection}, and \ref{lem:computeSumDetContainingFixedSet} in the proofs and then following a similar calculation. The proofs for running times are identical.
\end{proof}

\section{Integrality Gaps}

\subsection{Integrality Gap for $E$-Optimality}
\label{sec:alon-boppana}

Here we consider another objective for optimal design of experiments,
the $E$-optimal design objective, and show that our results in the
asymptotic regime do not extend to it. Once again, the input is a set
of vectors $v_1, \ldots, v_n \in \R^d$, and our goal is to select a
set $S\subseteq [n]$ of size $k$, but this time we minimize the
objective
$\bigl\|(\sum_{i \in S}{v_i v_i^\top})^{-1}\bigr\|$,
where $\| \cdot \|$ is the operator norm, i.e.~the largest singular
value. By taking the inverse of the objective, this is equivalent to maximizing \(\lambda_1(\sum_{i \in S}{v_i v_i^\top}) \), where \(\lambda_i(M)\) denotes the \(i\)th smallest eigenvalue of \(M\). This problem also has a natural convex relaxation, analogous to
the one we use for the $A$ objective:
\begin{align}
  &\max \lambda_1 \left(\sum_{i = 1}^n x_i v_i  v_i^\top \right) \label{eq:CP-E-obj}\\
  &\text{s.t.} \notag\\
  &\sum_{i = 1}^n{x_i} = k\\
  &0 \le x_i \le 1 \ \ \ \forall i \in [n]\label{eq:CP-E-bounds}
\end{align}
We prove the following integrality gap result for
\eqref{eq:CP-E-obj}--\eqref{eq:CP-E-bounds}.

\begin{theorem}\label{thm:E-int-gap}
  There exists a constant $c> 0$ such that the following holds. For
  any small enough $\epsilon > 0$, and all integers $d \ge
  d_0(\epsilon)$, if $k < \frac{cd}{\epsilon^2}$, then there exists an instance $v_1, \ldots v_n \in \R^d$ of the
  $E$-optimal design problem, for which the value $\CP$ of
  \eqref{eq:CP-E-obj}--\eqref{eq:CP-E-bounds} satisfies
  \[
  \CP > (1+\epsilon)\OPT = (1+\epsilon) \max_{S\subseteq [n]: |S| = k}\lambda_1 \left(\sum_{i \in S} v_i  v_i^\top \right) \]
\end{theorem}

Recall that for the $A$-objective we achieve a
$(1+\epsilon)$-approximation for $k = \Omega(\frac{d}{\epsilon} +
\frac{\log(1/\epsilon)}{\epsilon^2})$. Theorem~\ref{thm:E-int-gap}
shows that such a result is impossible for the $E$-objective, for
which the results in~\cite{Allen-Zhu17nearoptimal} cannot be improved.

Our integrality gap instance comes from a natural connection to
spectral graph theory. Let us first describe the instance for any
given $d$. We first define $n = {d+1 \choose 2}$ vectors in
$\R^{d+1}$, one for each unordered pair $(i, j) \in {[d+1] \choose
  2}$. The vector corresponding to $(i,j)$, $i < j$, is $u_{ij}$ and
has value $1$ in the $i$-th coordinate, $-1$ in the $j$-th coordinate,
and $0$ everywhere else. In other words, the $u_{ij}$ vectors are the
columns of the vertex by edge incidence matrix $U$ of the complete
graph $K_{d+1}$, and $UU^\top = (d+1)I_{d+1} - J_{d+1}$ is the
(unnormalized) Laplacian of $K_{d+1}$. (We use $I_m$ for the $m\times
m$ identity matrix, and $J_m$ for the $m\times m$ all-ones matrix.) All the $u_{ij}$ are orthogonal to
the all-ones vector $1$; we define our instance by writing $u_{ij}$ in
an orthonormal basis of this subspace: pick any orthonormal basis $b_1,
\ldots, b_d$ of the subspace of $\R^{d+1}$ orthogonal to $1$, and
define $v_{ij} = B^\top u_{ij}$ for $B = (b_i)_{i = 1}^d$. Thus
\[M = \sum_{i = 1}^{d+1} \sum_{j = i+1}^{d+1} v_{ij} v_{ij}^\top =
(d+1) I_d.\]
 We consider the fractional solution $x = \frac{k}{{d+1
    \choose 2}} 1$, i.e.~each coordinate of $x$ is
$k/{d+1\choose 2}$. Then $M(x) = \sum_{i = 1}^{d+1} \sum_{j = i+1}^{d+1} x_{ij}v_{ij} v_{ij}^\top= \frac{2k}{d} I_d$, and the
objective value of the solution is $\frac{2k}{d}$.

Consider now any integral solution $S \subseteq {[d+1]\choose 2}$ of
the $E$-optimal design problem. We can treat $S$ as the edges of a
graph $G = ([d+1], S)$, and the Laplacian $L_G$ of this graph is $L_G
= \sum_{(i,j) \in S}{u_{ij} u_{ij}^\top}$. If the objective value
of $S$ is at most $(1+\epsilon)\CP$, then the smallest eigenvalue of
$M(S) = \sum_{(i,j) \in S}{v_{ij} v_{ij}^\top}$ is at least
$\frac{2k}{d(1+\epsilon)} \ge (1-\epsilon) \frac{2k}{d}$. Since $M(S)
= B^\top L_G B$, this means that the second smallest eigenvalue
of $L_G$ is at least $(1-\epsilon) \frac{2k}{d}$. The average degree
$\Delta$ of $G$ is $\frac{2k}{d+1}$. So, we have a graph $G$ on $d+1$
vertices with average degree $\Delta$ for which the second smallest
eigenvalue of its Laplacian is at least $(1-\epsilon)(1
-\frac{1}{d+1}) \Delta \ge (1-2\epsilon)\Delta$, where the inequality
holds for $d$ large enough. The classical Alon-Boppana
bound~(\cite{Alon86,Nilli91}) shows that, up to lower order terms, the
second smallest eigenvalue of the Laplacian of a $\Delta$-regular graph
is at most $\Delta - 2\sqrt{\Delta}$. If our graph $G$ were regular,
this would imply that $\frac{2k}{d+1} = \Delta \ge
\frac{1}{\epsilon^2}$. In order to prove Theorem~\ref{thm:E-int-gap},
we extend the Alon-Boppana bound to not necessarily regular graphs,
but with worse constants. There is an extensive body of work on
extending the Alon-Boppana bound to non-regular graphs: see the recent
preprint~\cite{SrivastavaT17} for an overview of prior work on this
subject. However, most of the work focuses either on the normalized
Laplacian or the adjacency matrix of $G$, and we were unable to find
the statement below in the literature.

\begin{theorem}\label{thm:alon-boppana}
  Let $G = (V, E)$ be a graph with average degree $\Delta =
  \frac{2|E|}{|V|}$, and let $L_G$ be its unnormalized Laplacian
  matrix. Then, as long as $\Delta$ is large enough, and $|V|$ is
  large enough with respect to $\Delta$,
  \[
  \lambda_2(L_G) \le \Delta - c\sqrt{\Delta},
  \]
  where $\lambda_2(L_G)$ is the second smallest eigenvalue of $L_G$,
  and $c> 0$ is an absolute constant.
\end{theorem}
\begin{proof}
  By the variational characterization of eigenvalues, we need to find
  a unit vector $x$, orthogonal to $1$, such that $x^\top L_G x
  \le \Delta - c\sqrt{\Delta}$. Our goal is to use a vector $x$
  similar to the one used in the lower bound on the number of edges of
  a spectral sparsifier in~\cite{BSS12}. However, to apply this
  strategy we need to make sure that $G$ has a low degree vertex most
  of whose neighbors have low degree. This requires most of the work
  in the proof.

  So that we don't have to worry about making our ``test vector''
  orthogonal to $1$, observe that
  \begin{equation}
    \label{eq:test-vector}
    \lambda_2(L_G) = \min_{x \in \R^V}
    \frac{x^\top L_G  x}{x^\top x - (1^\top x)^2/|V|}.
  \end{equation}
  Indeed, the denominator equals $y^\top y$ for the projection
  $y$ of
  $x$ orthogonal to $1$, and the numerator is equal to $y^\top
  L_G y$. Here, and in the remainder of the proof, we work in $\R^V$,
  the space of $|V|$-dimensional real vectors indexed by $V$, and
  think of $L_G$ as being indexed by $V$ as well.

  Observe that if $G$ has a vertex $u$ of degree $\Delta(u)$ at most
  $\Delta - \frac{1}{10}\sqrt{\Delta}$, we are done. In that case we
  can pick $x \in \R^V$ such that $x_u = 1$ and $x_v = 0$ for all $v
  \neq u$. Then
  \[
  \frac{x^\top L_G  x}{x^\top x - (1^\top x)^2/n}
  = \frac{\sum_{(u,v) \in E}{(x_u - x_v)^2}}{1 - \frac{1}{|V|}}
  \le \frac{\Delta - \frac{1}{10}\sqrt{\Delta}}{1 - \frac{1}{|V|}},
  \]
  which, by \eqref{eq:test-vector}, implies the theorem for all large
  enough $|V|$. Therefore, for the rest of the proof we will assume
  that $\Delta(u) \ge \Delta - \frac{1}{10}\sqrt{\Delta}$ for all $u
  \in V$.

  Define $T = \{u \in V: \Delta(u) \ge \Delta +
  \frac{1}{2}\sqrt{\Delta}\}$ to be the set of large degree vertices,
  and let $S = V \setminus T$. Observe that
  \begin{align*}
  |V|\Delta &\ge |T|\bigl(\Delta + \frac{1}{2}\sqrt{\Delta}\bigr) + |S|\bigl(\delta -  \frac{1}{10}\sqrt{\Delta}\bigr)\\
  &=  |V|\Delta + \bigl(\frac{1}{2}|T| - \frac{1}{10}|S|\bigr)\sqrt{\Delta}.
  \end{align*}
  Therefore, $|S| \ge {5}|T|$, and, since $T$ and $S$
  partition $V$, we have $|S| \ge \frac{5}{6}|V|$.

  Define
  \[
  \alpha = \min\left\{\frac{|\{v \sim u: v \in T\}|}{\Delta - \frac{1}{10}\sqrt{\Delta}}: u \in S\right\},
  \]
  where $v\sim u$ means that $v$ is a neighbor of $u$.  We need to
  find a vertex in $S$ such that only a small fraction of its
  neighbors are in $T$, i.e.~we need an upper bound on $\alpha$. To
  show such an upper bound, let us define $E(S, T)$ to be the set of
  edges between $S$ and $T$; then
  \[
  \frac12 \Delta |V| = |E| \ge |E(S, T)|
  \ge |S| \alpha \left(\Delta - \frac{1}{10}\sqrt{\Delta}\right)
  \ge \frac{5}{6}|V| \alpha \Delta \left(1 - \frac{1}{10\sqrt{\Delta}}\right).
  \]
  Therefore, $\alpha \le \frac35 \bigl(1 - \frac{1}{10\sqrt{\Delta}}\bigr)^{-1}$.

  Let $u \in S$ be a vertex  with at most
  $\alpha \Delta - \frac{\alpha}{10}\sqrt{\Delta}$  neighbors in $T$, and let $\delta = |\{v \sim u:
  v \in S\}|$. By the choice of $u$,
  \[
  \delta \ge \Delta(u)-\alpha  \Delta + \frac{\alpha}{10}\sqrt{\Delta}
  \ge (1-\alpha)\Delta \Biggl(1  - \frac{1}{10\sqrt{\Delta}}\Biggr).
  \]
  Assume that $\Delta$ is large enough so that $\Biggl(1  -
  \frac{1}{10\sqrt{\Delta}}\Biggr) \ge \frac{16}{25}$. Then, $\delta \ge
  \frac{16}{25}(1-\alpha)\Delta$.

  We are now ready to define our test vector $x$ and complete the
  proof. Let $x_u = 1$, $x_v = \frac{1}{\sqrt{\delta}}$ for any neighbor $v$
  of $u$ which is in $S$, and $x_w = 0$ for any $w$ which is in $T$ or
  is not a neighbor of $u$. We calculate
  \begin{align*}
  x^\top L_G x &=
  |\{v \sim u: v \in S\}| \left(1 -  \frac{1}{\sqrt{\delta}}\right)^2
  +
  |\{v \sim u: v \in T\}|
  +
  \sum_{v \sim u, v \in S}\sum_{w \sim v, w \neq u} \frac{1}{\delta}\\
  &\le
  \delta \left(1 -  \frac{1}{\sqrt{\delta}}\right)^2
  + \Delta(u) - \delta
  + \Delta + \frac12 \sqrt{\Delta} - 1,
  \end{align*}
  where we used the fact for any $v \in S$, $\Delta(v) \le \Delta +
  \frac12 \sqrt{\Delta}$ by definition of $S$. The right hand side simplifies
  to
  \[
  \Delta(u) - 2\sqrt{\delta} + \Delta + \frac12 \sqrt{\Delta}
  \le
  2\Delta - \left(\frac{8}{5}\sqrt{(1-\alpha)} - \frac12\right)\sqrt{\Delta} .
  \]
  Since $\alpha \le \frac35\bigl(1 -
  \frac{1}{10\sqrt{\Delta}}\bigr)^{-1}$, $\frac{8}{5}\sqrt{(1-\alpha)} - \frac12
  \ge \frac12$ for all large enough $\Delta$, and by
  \eqref{eq:test-vector}, we have
  \[
  \lambda_2(G) \le
  \frac{x^\top L_G x}{x^\top x - (1^\top x)^2}
  \le \frac{2\Delta - \frac12 \sqrt{\Delta}}{2\left(1 - \frac{1 +
        \sqrt{\Delta}}{2|V|}\right)}
  = \left(\Delta - \frac14 \sqrt{\Delta}\right)
  \left(1 - \frac{1 + \sqrt{\Delta}}{2|V|}\right)^{-1}.
  \]
  The theorem now follows as long as $|V| \ge C \Delta$ for a
  sufficiently large constant $C$.
\end{proof}

To finish the proof of Theorem~\ref{thm:E-int-gap}, recall that the
existence of a $(1+\epsilon)$-approximate solution $S$ to our instance
implies that, for all large enough $d$, the graph $G = ([d+1], S)$
with average degree $\Delta = \frac{2k}{d+1}$ satisfies
$\lambda_2(L_G) \ge (1 - 2\epsilon)\Delta$. By
Theorem~\ref{thm:alon-boppana}, $\lambda_2(L_G) \le \Delta -
c\sqrt{\Delta}$  for large enough $d$ with respect to $\Delta$. We
have $\Delta \ge \frac{c^2}{4\epsilon^2}$, and re-arranging the terms
proves the theorem.

Note that the proof of Theorem~\ref{thm:alon-boppana} does not require
the graph $G$ to be simple, i.e.~parallel edges are allowed. This
means that the integrality gap in Theorem~\ref{thm:E-int-gap} holds
for the $E$-optimal design problem with repetitions as well.

\subsection{Integrality Gap for $A$-optimality}\label{sec:integralityA}
\begin{theorem} \label{thm:IntegralityGapLowerBound}
For any given positive integers \(k,d \), there exists an instance \(V=[v_1,\ldots,v_n]\in\R^{d \times n}\) to the \(A\)-optimal design problem such that 
\[ \OPT \geq \left(\frac{k}{k-d+1}-\delta\right)\cdot \CP \]
for all $\delta >0,$ where $\OPT$ denotes the value of the optimal integral solution and $\CP$ denotes the value of the convex program.\end{theorem}
This implies that the gap is at least $\frac{k}{k-d+1}$. The theorem statement applies to both with and without repetitions.

\begin{proof} 
The instance $V=[v_1,\ldots,v_n]$ will be the same  with or without repetitions. For each \(1\leq i \leq d\), let \(e_i\) denote the unit vector in direction of axis \(i\).  
Let $v_i = N\cdot e_i$ for each $i=1,\ldots,d-1$, where \(N>0\) is a constant to be chosen later and $v_d=e_d$. Set the rest $v_i,i>d$ to be at least $k$ copies of each of these $v_i$ for $i\leq d$, as we can make $n$ as big as needed. Hence, we may assume that we are allowed to pick only $v_i,i\leq d$, but with repetitions.

The fractional optimal solution which can be calculated by Lagrange's multiplier technique is $y^* = (\delta_0,\delta_0,\ldots,\delta_0,k-(d-1)\delta_0)$ for small $\delta_0 = \frac{k}{\sqrt{N}+d-1}$. The optimal integral solution is $x^*=(1,1,\ldots,1,k-d+1)$. Therefore, as $N\rightarrow \infty$, we have $\CP = \frac{d-1}{\delta_0 N}+\frac{1}{k-(d-1)\delta_0} \rightarrow \frac{1}{k}$, and $\OPT = \frac{d-1}{N} + \frac{1}{k-d+1} \rightarrow \frac{1}{k-d+1}$. Hence,
\[\frac{\OPT}{\CP} \rightarrow \frac{k}{k-d+1}\]
proving the theorem.
\end{proof}

\section{Hardness of Approximation}

In this section we show that the $A$-optimal design problem is
$\NP$-hard to approximate within a fixed constant when $k=d$. To the
best of our knowledge, no hardness results for this problem were previously known.
Our reduction is inspired by the hardness of approximation for
$D$-optimal design proved in~\cite{summa2015largest}. The hard problem
we reduce from is an approximation version of Partition into Triangles.

Before we prove our main hardness result, Theorem~\ref{thm:hardness},
we describe the class of instances we consider, and prove some basic
properties. Given a graph $G = ([d], E)$, we define a vector $v_e$ for
each edge $e = (i, j)$ so that its $i$-th and $j$-th coordinates are
equal to $1$, and all its other coordinates are equal to $0$. Then the
matrix $V = (v_e)_{e \in E}$ is the undirected vertex by edge
incidence matrix of $G$. The main technical lemma needed for our
reduction follows.

\begin{lemma}\label{lm:triangles}
  Let $V$ be the vertex by edge incidence matrix of a graph $G = ([d],
  E)$, as described above. Let $S \subseteq E$ be a set of $d$ edges
  of $G$ so that the submatrix $V_S$ is invertible. Then each
  connected component of the subgraph $H = ([d], S)$ is the disjoint
  union of a spanning tree and an edge. Moreover, if $t$ of the
  connected components of $H$ are triangles, then
  \begin{itemize}
  \item for $t = \frac{d}{3}$, $\tr((V_SV_S^\top)^{-1}) = \frac{3d}{4}$;
  \item for any $t$, $\tr((V_SV_S^\top)^{-1}) \ge d - \frac{3t}{4}$.
  \end{itemize}
\end{lemma}
\begin{proof}
  Let $H_1, \ldots, H_c$ be the connected components of $H$. First we
  claim that the invertibility of $V_S$ implies that none of the $H_\ell$
  is bipartite. Indeed, if some
  $H_\ell$ were bipartite, with bipartition $L \cup R$, then the nonzero vector
  $x$ defined by
  \[
  x_i =
  \begin{cases}
    1 &i \in L\\
    -1 &i \in R\\
    0 &\text{otherwise},
  \end{cases}
  \]
  is in the kernel of $V_S$. In particular, each $H_\ell$ must have at
  least as many edges as vertices. Because the number of edges of $H$
  equals the number of vertices, it follows that \emph{every}
  connected component $H_\ell$ must have exactly as many edges as
  vertices, too. In particular, this means that every $H_\ell$ is the
  disjoint union of a spanning tree and an edge, and the edge creates
  an odd-length cycle.


  Let us explicitly describe the inverse $V_S^{-1}$.  For each $e \in
  S$ we need to give a vector $u_e \in \R^d$ so that $u_e^\top v_e =
  1$ and $u_e^\top v_f = 0$ for every $f \in S$, $f \neq e$. Then
  $U^\top = V_S^{-1}$, where $U = (u_e)_{e \in S}$ is the matrix whose
  columns are the $u_e$ vectors.  Let $H_\ell$ be, as above, one of
  the connected components of $H$. We will define the vectors $u_e$
  for all edges $e$ in $H_\ell$; the vectors for edges in the other
  connected components are defined analogously. Let $C_\ell$ be
  the unique cycle of $H_\ell$. Recall that $C_\ell$ must be an odd
  cycle. For any $e = (i,j)$ in $C_\ell$, we set the $i$-th and the
  $j$-th coordinate of $u_e$ to $\frac{1}{2}$. Let $T$ be the spanning
  tree of $H_\ell$ derived from removing the edge $e$.  We set the
  coordinates of $u_e$ corresponding to vertices of $H_\ell$ other
  than $i$ and $j$ to either $-\frac12$ or $+\frac12$, so that the
  vertices of any edge of $T$ receive values with opposite signs. This
  can be done by setting the coordinate of $u_e$ corresponding to
  vertex $k$ in $H_\ell$ to $\frac12 (-1)^{\delta_T(i, k)}$, where
  $\delta_T(i,k)$ is the distance in $T$ between $i$ and $k$. Because
  $C_\ell$ is an odd cycle, $\delta_T(i,j)$ is even, and this
  assignment is consistent with the values we already determined for
  $i$ and $j$. Finally, the coordinates of $u_e$ which do not
  correspond to vertices of $H_\ell$ are set to $0$. See
  Figure~\ref{fig:inverse-cycle} for an example. It is easy to verify
  that $u_e^\top v_e = 1$ and $u_e^\top v_f = 0$ for any edge $f \neq
  e$. Notice that $\|u_e\|_2^2 = \frac{d_\ell}{4}$, where $d_\ell$ is
  the number of vertices (and also the number of edges) of $H_\ell$.

  \begin{figure}[htp]
    \centering
    \includegraphics{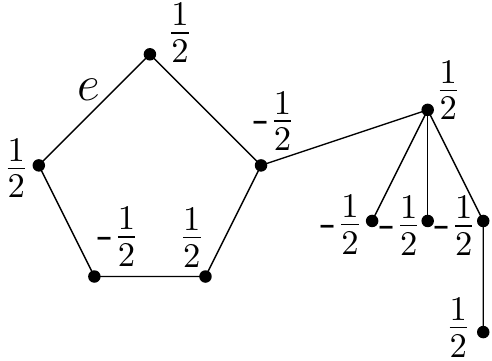}
    \caption{The values of the coordinates of $u_e$ for $e \in C_\ell$.}
    \label{fig:inverse-cycle}
  \end{figure}

  It remains to describe $u_e$ when $e = (i,j) \not \in C_\ell$. Let
  $T$ be the tree derived from $H_\ell$ by contracting $C_\ell$ to a
  vertex $r$, and set $r$ as the root of $T$. Without loss of
  generality, assume that $j$ is the endpoint of $e$ which is further
  from $r$ in $T$. We set the $j$-th coordinate of $u_e$ equal to
  $1$. We set the coordinates of $u_e$ corresponding to vertices in
  the subtree of $T$ below $j$ to either $-1$ or $+1$ so that the
  signs alternate down each path from $j$ to a leaf of $T$ below
  $j$. This can be achieved by setting the coordinate of $u_e$
  corresponding to vertex $k$ to $(-1)^{\delta_T(j,k)}$, where
  $\delta_T(j,k)$ is the distance between $j$ and $k$ in $T$. All
  other coordinates of $u_e$ are set equal to $0$. See
  Figure~\ref{fig:inverse-tree} for an example. Notice that
  $\|u_e\|_2^2 \ge 1$ (and in fact equals the number of nodes in the
  subtree of $T$ below the node $j$).

  \begin{figure}[htp]
    \centering
    \includegraphics{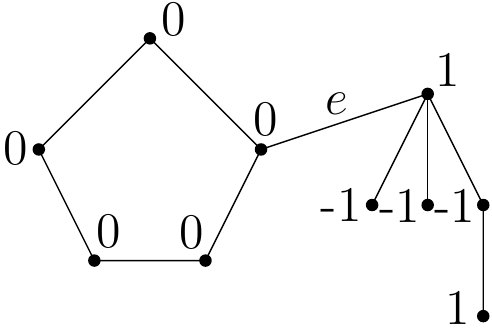}
    \caption{The values of the coordinates of $u_e$ for $e \not\in C_\ell$.}
    \label{fig:inverse-tree}
  \end{figure}

  We are now ready to finish the proof. Clearly if $[d]$ can be
  partitioned into $t = \frac{d}{3}$ disjoint triangles, and the union of their
  edges is $S$, then
  \[
  \tr((V_S V_S^\top)^{-1}) = \tr(U U^\top) = \sum_{e \in
    S}{\|u_e\|_2^2} = \frac{3|S|}{4} = \frac{3d}{4}.
  \]
  In the general case, we have
  \begin{align*}
  \tr((V_S V_S^\top)^{-1}) = tr(U U^\top)
  &= \sum_{e \in S}{\|u_e\|_2^2} \\
  &\ge \sum_{\ell = 1}^c{\frac{|C_\ell|\cdot d_\ell}{4} + d_\ell - |C_\ell|}\\
  &\ge \frac{9t}{4} + d-3t = d - \frac{3t}{4},
  \end{align*}
  where $|C_\ell|$ is the length of $C_\ell$, and $d_\ell$ is the
  number of edges (and also the number of vertices) in $H_\ell$. The
  final inequality follows because any connected component $H_\ell$
  which is not a triangle contributes at least $d_\ell$ to the sum.
\end{proof}

Recall that in the Partition into Triangles problem we are given a
graph $G = (W, E)$, and need to decide if $W$ can be partitioned
into $\frac{|W|}{3}$ vertex-disjoint triangles. This problem is
$\NP$-complete (\cite{GJ79} present a proof in
Chapter 3 and cite personal communication with Schaeffer), and this,
together with Lemma~\ref{lm:triangles}, suffice to show that the
$A$-optimal design problem is $\NP$-hard when $k=d$. To prove hardness
of approximation, we prove hardness of a gap version of Partition into
Triangles. In fact, we just observe that the reduction from
3-Dimensional Matching to Partition into Triangles in~\cite{GJ79} and
known hardness of approximation of 3-Dimensional Matching give the
result we need.

\begin{lemma}\label{lm:tri-gap}
  Given a graph $G = (W, E)$, it is $\NP$-hard to distinguish the
  two cases:
  \begin{enumerate}
  \item $W$ can be partitioned into $\frac{|W|}{3}$ vertex-disjoint triangles;
  \item every set of vertex-disjoint triangles in $G$ has cardinality
    at most $\alpha\frac{|W|}{3}$,
  \end{enumerate}
  where  $\alpha \in (0,1)$ is an absolute constant.
\end{lemma}

To prove Lemma~\ref{lm:tri-gap} we use a theorem of Petrank.

\begin{theorem}[\cite{Petrank94}]\label{thm:3dm-gap}
  Given a collection of triples $F\subseteq X \times Y \times Z$, where
  $X$, $Y$, and $Z$ are three disjoint sets of size $m$ each, and each
  element of $X \cup Y \cup Z$ appears in at most $3$ triples of $F$, it is
  $\NP$-hard to distinguish the two cases
  \begin{enumerate}
  \item there is a set of disjoint triples $M \subseteq F$ of
    cardinality $m$;
  \item every set of disjoint triples $M \subseteq F$ has cardinality
    at most $\beta m$,
  \end{enumerate}
  where $\beta \in (0,1)$ is an absolute constant.
\end{theorem}

We note that Petrank gives a slightly different version of the
problem, in which the set $M$ is allowed to have intersecting triples,
and the goal is to maximize the number of elements $X \cup Y \cup Z$
that are covered exactly once. Petrank shows that it is hard to
distinguish between the cases when every element is covered exactly
once, and the case when at most $3\beta m $ elements are covered
exactly once. It is immediate that this also implies
Theorem~\ref{thm:3dm-gap}.

\medskip
\begin{proofof}{Lemma~\ref{lm:tri-gap}}
  We will show that the reduction in~\cite{GJ79} from 3-Dimensional
  Matching to Partition into Triangles is approximation
  preserving. This follows in a straightforward way from the argument
  in~\cite{GJ79}, but we repeat the reduction and its analysis for the
  sake of completeness.

  Given $F \subseteq X \cup Y \cup Z$ such that each element of $X
  \cup Y \cup Z$ appears in at most $3$ tripes of $F$, we construct a
  graph $G = (W, E)$ on the vertices $X \cup Y \cup Z$ and $9|F|$
  additional vertices: $a_{f1}, \ldots a_{f9}$ for each $f \in F$. For
  each triple $f \in F$, we include in $E$ the edges $E_f$ shown in
  Figure~\ref{fig:tri-gadget}. Note that the subgraphs spanned by the sets $E_f$, $E_g$ for
  two different triples $f$ and $g$ are edge-disjoint, and the only vertices they
  share are in $X \cup Y \cup Z$.

  \begin{figure}[htp]
    \centering
    \includegraphics{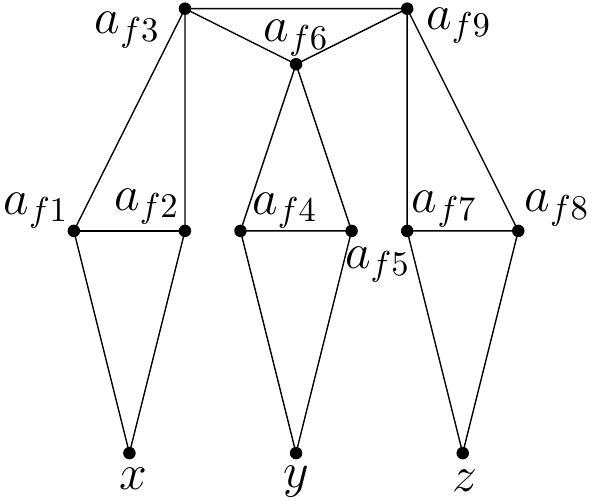}
    \caption{The subgraph with edges $E_f$ for the triple $f = \{x, y, z\}$. (Adapted from~\cite{GJ79})}
    \label{fig:tri-gadget}
  \end{figure}

  First we show that if $F$ has a matching $M$ covering all
  elements of $X \cup Y \cup Z$, then $G$ can be partitioned into
  vertex-disjoint triangles. Indeed, for each $f = \{x,y,z\} \in M$ we can take
  the triangles $\{x, a_{f1}, a_{f2}\}$, $\{y, a_{f4}, a_{f5}\}$,
  $\{z, a_{f7}, a_{f8}\}$, and $\{a_{f3}, a_{f6}, a_{f9}\}$. For each
  $f \not \in M$ we can take the triangles $\{a_{f1}, a_{f2},
  a_{f3}\}$, $\{a_{f4}, a_{f5}, a_{f6}\}$, and $\{a_{f7}, a_{f8},
  a_{f9}\}$.

  In the other direction, assume there exists a set $T$ of at least
  $\alpha \frac{|W|}{3}$ vertex disjoint triangles in $G$, for a value
  of $\alpha$ to be chosen shortly. We need to show that $F$ contains
  a matching of at least $\beta m$ triples. To this end, we construct
  a set $M$ which contains all triples $f$, for each $E_f$ which
  contains at least $4$ triangles of $T$. Notice that the only way to
  pick three vertex disjoint triangles from $E_f$ is to include the
  lower three triangles (see Figure), so any two triples $f$ and $g$
  in $M$ must be disjoint. The cardinality of $T$ is at most $4|M| +
  3(|F| - |M|) = |M| + 3|F|$. Therefore,
  \[
  |M| + 3|F| \ge \alpha \frac{|W|}{3} = \alpha (m + 3|F|),
  \]
  and we have $|M| \ge \alpha m - (1-\alpha)3|F| \ge (10\alpha - 9)m$,
  where we used the fact that $|F| \le 3m$ because each element of $X$
  appears in at most $3$ triples of $F$. Then, if $\alpha \ge \frac{9
    + \beta}{10}$ we have $|M| \ge \beta m$. This finishes the proof
  of the lemma.
\end{proofof}

We now have everything in place to finish the proof of our main
hardness result.

\medskip
\begin{proofof}{Theorem~\ref{thm:hardness}}
  We use a reduction from (the gap version of) Partition into
  Triangles to the $A$-optimal design problem. In fact the reduction
  was already described in the beginning of the section: given a graph
  $G = ([d], E)$, it outputs the columns $v_e$ of the vertex by edge
  incidence matrix $V$ of $G$.

  Consider the case in which the vertices of $G$ can be partitioned into
  vertex-disjoint triangles. Let $S$ be the union of the edges of the
  triangles. Then, by Lemma~\ref{lm:triangles},
  $\tr((V_SV_S^\top)^{-1}) = \frac{3d}{4}$.

  Next, consider the case in which every set of vertex-disjoint
  triangles in $G$ has cardinality at most $\alpha\frac{d}{3}$. Let
  $S$ be any set of $d$ edges in $E$ such that $V_S$ is
  invertible. The subgraph $H = ([d], S)$ of $G$ can have at most
  $\alpha \frac{d}{3}$ connected components that are triangles, because any
  two triangles in distinct connected components are necessarily
  vertex-disjoint. Therefore, by Lemma~\ref{lm:triangles},
  $\tr((V_SV_S^\top)^{-1}) \ge \frac{(4 -\alpha) d}{4}$.

  It follows that a $c$-approximation algorithm for the $A$-optimal
  design problem, for any $c < \frac{4-\alpha}{3}$, can be used to
  distinguish between the two cases of Lemma~\ref{lm:tri-gap}, and,
  therefore, the $A$-optimal design problem is $\NP$-hard to
  $c$-approximate.
\end{proofof}


\end{document}